\documentclass{article}%
\usepackage{amsmath}
\usepackage{amsfonts}
\usepackage{amssymb}
\usepackage{graphicx}%
\providecommand{\U}[1]{\protect\rule{.1in}{.1in}}
\newtheorem{theorem}{Theorem}

\newtheorem{definition}[theorem]{Definition}

\newtheorem{proposition}[theorem]{Proposition}
\newtheorem{remark}[theorem]{Remark}

\newenvironment{proof}[1][Proof]{\noindent\textbf{#1.} }{\ \rule{0.5em}{0.5em}}
\begin{document}

\title{Concept of Lie Derivative of Spinor Fields. A Geometric Motivated Approach}
\author{Waldyr A. Rodrigues Jr., Rafael F. Le\~{a}o and Samuel A. Wainer\\IMECC-UNICAMP\\{\footnotesize \emph{~}walrod@ime.unicamp.br~~leao@ime.unicamp.br\emph{~}%
~samuelwainer@ime.unicamp.br}}
\date{November 30 2015}
\maketitle
\tableofcontents

\begin{abstract}
In this paper using the Clifford bundle ($\mathcal{C\ell}(M,\mathtt{g})$) and
spin-Clifford bundle ($\mathcal{C\ell}_{\mathrm{Spin}_{1,3}^{e}}%
(M,\mathtt{g})$) formalism, which permit to give a meaningfull representative
of a Dirac-Hestenes spinor field (even section of $\mathcal{C\ell
}_{\mathrm{Spin}_{1,3}^{e}}(M,\mathtt{g})$) in the Clifford bundle , we give a
geometrical motivated definition for the Lie derivative of spinor fields in a
Lorentzian structure $(M,\boldsymbol{g})$ where $M$ is a manifold such that
$\dim M=4$, $\boldsymbol{g}$ is Lorentzian of signature $(1,3)$. Our Lie
derivative, called the spinor Lie derivative (and denoted
$\overset{s}{\pounds }_{\boldsymbol{\xi}}$) is given by nice formulas when
applied to Clifford and spinor fields, and moreover $\overset{s}{\pounds }%
_{\boldsymbol{\xi}}\boldsymbol{g}=0$ for any vector field $\boldsymbol{\xi}$.
We compare our definitions and results with the many others appearing in
literature on the subject.

\end{abstract}

\section{Introduction}

Our principal aim in this paper is, using the Clifford and spin-Clifford
bundles formalism, to give a geometrical motivated definition for the Lie
derivative of spinor fields in a Lorentzian structure $(M,\boldsymbol{g})$,
that will be defined below.

In Section 2 we recall some key definitions and some propositions that will
suggest us how to define the\ spinor image\ of Clifford and spinor fields
using the spinor lifting of an integral curve of a vector field, as set in
Definition (\ref{SPLIFT}). In section 3 we expose the main problem concerning
the definition of an appropriate definition for the Lie derivative of spinor
fields. We them present a proposition that permit us to calculate the usual
Lie derivative ($\pounds _{\boldsymbol{\xi}}$) of a coterad basis in the
direction of an arbitrary vector field $\boldsymbol{\xi}$ in two different
ways, the usual way, and one of them making use of the concept of the spinor
lifting (Definition \ref{SPLIFT}) of an integral curve of $\boldsymbol{\xi}$
in $P_{Spin_{1,3}}(M,\mathtt{g})$. It is this way of obtaining the Lie
derivative of a cotetrad basis that suggested us to give in Section 4.1 a
geometrical motivated definition of spinor images of Clifford and spinor
fields introducing the spinor mapping $^{s}$\textrm{h}$_{t}$ and next to
define in Section 4.2 the \emph{spinor Lie derivative }(denoted
$\overset{s}{\pounds _{\boldsymbol{\xi}}}$)\emph{ }of Clifford and spinor
fields and them to calculate the explicit simple and nice formulas for those
objects. The idea of Section 5\ is to write the spinor Lie derivative in terms
of covariant derivatives in such way that we can relate our construction with
the literature, as it appears in \cite{lichnerowicz}. In particular we
evaluate in Section 5.1 the spinor Lie derivative of a representative of a
DHSF in local coordinates and in Section 5.2 we write the spinor Lie
derivative for covariant Dirac spinor fields. In Section 6 we show that
$\overset{s}{\pounds _{\boldsymbol{\xi}}}\boldsymbol{g}=0$, for any arbitrary
vector field $\boldsymbol{\xi}$. A definition of Lie derivative that
annihilates $\boldsymbol{g}$ has been given firstly by Bourguignon \&
Gauduchon \cite{bg1992}, but our approach is very different from the one used
by those authors. The main proposal of Section 7 is to comment on some
different approaches to the Lie derivative of spinor fields, with conflicting
views appearing in the literature, and how our geometrical approach intersects
these and present our future prospects. Finally, in Section 8 we present our conclusions.

\section{Preliminaries}

Here, $M$ refers\footnote{Unless, explicitly stated.} to a four dimensional,
real, connected, paracompact and non-compact manifold. We define a Lorentzian
manifold as a pair $(M,\boldsymbol{g})$, where $\boldsymbol{g}\in\sec
T_{2}^{0}M$ is a Lorentzian metric of signature $(1,3)$, i.e., $\forall x\in
M,T_{x}M\simeq T_{x}^{\ast}M\simeq\mathbb{R}^{1,3}$, where $\mathbb{R}^{1,3}$
is the Minkowski vector space. We define a Lorentzian spacetime $M$ as
pentuple $(M,\boldsymbol{g},\boldsymbol{D},\tau_{\boldsymbol{g}},\uparrow)$,
where $(M,\boldsymbol{g},\tau_{\boldsymbol{g}},\uparrow)$) is an oriented
Lorentzian manifold (oriented by $\tau_{\boldsymbol{g}}$) and time oriented by
$\uparrow$, and $\boldsymbol{D}$ is the Levi-Civita connection of
$\boldsymbol{g}$. Let $\mathcal{U}\subseteq M$ be an open set covered by
coordinates $\{x^{\mu}\}$. Let $\{e_{\mu}=\partial_{\mu}\}$ be a coordinate
basis of $T\mathcal{U}$ and $\{\boldsymbol{\vartheta}^{\mu}=dx^{\mu}\}$ the
dual basis on $T^{\ast}\mathcal{U}$, i.e., $\boldsymbol{\vartheta}^{\mu
}(\partial_{\nu})=\delta_{\nu}^{\mu}$. If $\boldsymbol{g}=g_{\mu\nu
}\boldsymbol{\vartheta}^{\mu}\otimes\boldsymbol{\vartheta}^{\nu}$ is the
metric on $T\mathcal{U}$ we denote by $\mathtt{g}=g^{\mu\nu}%
\boldsymbol{\partial}_{\mu}\otimes\boldsymbol{\partial}_{\nu}$ the metric of
$T^{\ast}\mathcal{U}$, such that $g^{\mu\rho}g_{\rho\nu}=\delta_{\nu}^{\mu}$.
We introduce also $\{\boldsymbol{\partial}^{\mu}\}$ and
$\{\boldsymbol{\vartheta}_{\mu}\}$, respectively, as the reciprocal bases of
$\{e_{\mu}\}$ and $\{\boldsymbol{\vartheta}_{\mu}\}$, i.e., we have
\begin{equation}
\boldsymbol{g}(\boldsymbol{\partial}_{\nu},\boldsymbol{\partial}^{\mu}%
)=\delta_{\nu}^{\mu},~~~\mathtt{g}(\boldsymbol{\vartheta}^{\mu}%
,\boldsymbol{\vartheta}_{\nu})=\delta_{\nu}^{\mu}. \label{p1}%
\end{equation}

In what follows $\mathbf{P}_{\mathrm{SO}_{1,3}^{e}}(M,\boldsymbol{g})$
($P_{\mathrm{SO}_{1,3}^{e}}(M,\mathtt{g})$) denotes the principal bundle of
oriented Lorentz tetrads (cotetrads).

\begin{definition}
A spin structure for a general $m$-dimensional manifold $M$ consists of a
principal fiber bundle $\pi_{s}:P_{\mathrm{Spin}_{p,q}^{e}}(M,\mathtt{g}%
)\rightarrow M$, \emph{(}called the Spin Frame Bundle\emph{)} with group
$\mathrm{Spin}_{p,q}^{e}$ and a map%
\begin{equation}
\Lambda:P_{\mathrm{Spin}_{p,q}^{e}}(M,\mathtt{g})\rightarrow P_{\mathrm{SO}%
_{p,q}^{e}}(M,\mathtt{g}),
\end{equation}
satisfying the following conditions:

\begin{description}
\item[(i)] $\pi(\Lambda(p))=\pi_{s}(p),\forall p\in P_{\mathrm{Spin}_{p,q}%
^{e}}(M,\mathtt{g})$, where $\pi$ is the projection map of the bundle
$\pi:P_{\mathrm{SO}_{p,q}^{e}}(M,\mathtt{g})\rightarrow M$.

\item[(ii)] $\Lambda(pu)=\Lambda(p)\mathrm{Ad}_{u},\forall p\in
P_{\mathrm{Spin}_{p,q}^{e}}(M,\mathtt{g})$ and $\mathrm{Ad}:\mathrm{Spin}%
_{p,q}^{e}\rightarrow\mathrm{SO}_{p,q}^{e}, \break\mathrm{Ad}_{u}%
(a)=uau^{-1}.$
\end{description}
\end{definition}

\begin{definition}
Any section of $P_{\mathrm{Spin}_{p,q}^{e}}(M,\mathtt{g})$ is called a spin
frame field \emph{(}or simply a spin frame\emph{)}. We shall use the symbol
$\Xi\in\sec P_{\mathrm{Spin}_{p,q}^{e}}(M,\mathtt{g})$ to denoted a spin frame.
\end{definition}

In this work we will assume that exists a spin structure on \ the
4-dimensional Lorentzian manifold $(M,\boldsymbol{g})$, what implies that $M$
is parallelizable, i.e., $P_{\mathrm{SO}_{1,3}^{e}}(M,\mathtt{g})$ is trivial,
because of the following result:

\begin{theorem}
\label{gerochh}For a $4-$dimensional Lorentzian manifold$\ (M,\boldsymbol{g}%
)$, a spin structure exists if and only if $P_{\mathrm{SO}_{1,3}^{e}%
}(M,\mathtt{g})$ is a trivial bundle.
\end{theorem}

\begin{proof}
See Geroch \cite{geroch}
\end{proof}

The Clifford bundle of differential forms of a Lorentzian manifold
$(M,\boldsymbol{g})$ is the bundle of algebras $\mathcal{C}\ell(M,\mathtt{g}%
)=\bigsqcup\nolimits_{x\in M}\mathcal{C}\ell(T_{x}^{\ast}M,\mathtt{g}_{x}).$

We know that\footnote{Where $Ad:\mathrm{Spin}_{1,3}^{e}\rightarrow
\mathrm{End(}\mathbb{R}_{1,3}\mathrm{)}$\textrm{ } is such that
$Ad(u)a=uau^{-1}$. And $\rho:\mathrm{SO}_{1,3}^{e}\rightarrow\mathrm{End(}%
\mathbb{R}_{1,3}\mathrm{)}$ is the natural action of $\mathrm{SO}_{1,3}^{e}$
on $\mathbb{R}_{1,3}$.} \cite{rc2007}:%
\begin{equation}
\mathcal{C}\ell(M,\mathtt{g})=P_{\mathrm{SO}_{1,3}^{e}}(M,\mathtt{g}%
)\times_{\rho}\mathbb{R}_{1,3}=P_{\mathrm{Spin}_{1,3}^{e}}(M,\mathtt{g}%
)\times_{Ad}\mathbb{R}_{1,3}, \label{eq_cliff}%
\end{equation}
and since\footnote{Given the objets $A$ and $B$, $A$ $\hookrightarrow$ $B$
means as usual that $A$ is embedded in $B$ and moreover, $A\subseteq B$. In
particular, recall that there is a canonical vector space isomorphism between
$\bigwedge\mathbb{R}^{1,3}$ and $\mathbb{R}_{1,3}$, which is written
$\bigwedge\mathbb{R}^{1,3}\hookrightarrow\mathbb{R}_{1,3}$. Details in
\cite{cru,lawmi}.} $\bigwedge TM\hookrightarrow\mathcal{C}\ell(M,\mathtt{g})$,
sections of $\mathcal{C}\ell(M,\mathtt{g})$ (the Clifford fields) can be
represented as a sum of non homogeneous differential forms.

Next (using that $M$ is parallelizable) we introduce the global tetrad basis
$\{\boldsymbol{e}_{\alpha}\}$ on $TM$ and in $T^{\ast}M$ the cotetrad basis on
$\{\boldsymbol{\gamma}^{\alpha}\}$, which are dual basis. We introduce the
reciprocal basis $\{\boldsymbol{e}^{\alpha}\}$ and $\{\boldsymbol{\gamma
}_{\alpha}\}$ of $\{\boldsymbol{e}_{\alpha}\}$ and $\{\boldsymbol{\gamma
}^{\alpha}\}$ satisfying
\begin{equation}
\boldsymbol{g}(\boldsymbol{e}_{\alpha},\boldsymbol{e}^{\beta})=\delta_{\alpha
}^{\beta},~~~\mathtt{g}(\boldsymbol{\gamma}^{\beta},\boldsymbol{\gamma
}_{\alpha})=\delta_{\alpha}^{\beta}. \label{p2}%
\end{equation}

Moreover, recall that\footnote{Where the matrix with entries $\eta
_{\alpha\beta}$ (or $\eta^{\alpha\beta}$) is the diagonal matrix
$(1,-1,-1,-1)$.}%
\begin{equation}
\boldsymbol{g}=\eta_{\alpha\beta}\boldsymbol{\gamma}^{\alpha}\otimes
\boldsymbol{\gamma}^{\beta}=\eta^{\alpha\beta}\boldsymbol{\gamma}_{\alpha
}\otimes\boldsymbol{\gamma}_{\beta},~~~\mathtt{g}=\eta^{\alpha\beta
}\boldsymbol{e}_{\alpha}\otimes\boldsymbol{e}_{\beta}=\eta_{\alpha\beta
}\boldsymbol{e}^{\alpha}\otimes\boldsymbol{e}^{\beta}. \label{P3}%
\end{equation}

To present our results on the Lie derivatives of spinor fields we need to
recall some other definitions, which serve also to fix our notation:

Recalling that $\mathrm{Spin}_{1,3}^{e}\hookrightarrow\mathbb{R}_{1,3}^{0}$,
we give:

\begin{definition}
The left \emph{(}respectively right\emph{)} real spin-Clifford bundle of the
spin manifold $M$ is the vector bundle $\mathcal{C}\ell_{\mathrm{Spin}}%
^{l}(M,\mathtt{g})=P_{\mathrm{Spin}_{1,3}^{e}}(M,\mathtt{g})\times
_{l}\mathbb{R}_{1,3}$ \emph{(}respectively $\mathcal{C}\ell_{\mathrm{Spin}%
}^{r}(M,\mathtt{g})=P_{\mathrm{Spin}_{1,3}^{e}}(M,\mathtt{g})\times
_{r}\mathbb{R}_{1,3}$\emph{)} where $l$ is the representation of
$\mathrm{Spin}_{1,3}^{e}$ on $\mathbb{R}_{1,3}$ given by $l(a)x=ax$
\emph{(}respectively, where $r$ is the representation of $\mathrm{Spin}%
_{1,3}^{e}$ on $\mathbb{R}_{1,3}$ given by $r(a)x=xa^{-1}$\emph{)}. Sections
of $\mathcal{C}\ell_{\mathrm{Spin}}^{l}(M,\mathtt{g})$ are called left
spin-Clifford fields \emph{(}respectively right spin-Clifford fields\emph{)}.
\end{definition}

\begin{definition}
Let $\mathbf{e^{l}}\in\mathcal{C}\ell_{\mathrm{Spin}_{1,3}^{e}}^{l}%
(M,\mathtt{g})$ be a primitive global idempotent\footnote{We know that global
primitive idempotents exist because $M$ is parallelizable.}, respectively
$\mathbf{e^{r}}\in\mathcal{C}\ell_{\mathrm{Spin}_{1,3}^{e}}^{r}(M,\mathtt{g}%
)$, and let $I^{l}(M,\mathtt{g})$\ and $I^{r}(M,\mathtt{g})$ be the subbundles
of $\mathcal{C}\ell_{\mathrm{Spin}_{1,3}^{e}}^{l}(M,\mathtt{g})$ and
$\mathcal{C}\ell_{\mathrm{Spin}_{1,3}^{e}}^{r}(M,\mathtt{g})$ generated by
these idempotents, that is, if $\mathbf{\Psi}$ is a section of $I^{l}%
(M,\mathtt{g})\subset\mathcal{C}\ell_{\mathrm{Spin}_{1,3}^{e}}^{l}%
(M,\mathtt{g})$, and $\mathbf{\Phi}$ is a section of $I^{r}(M,\mathtt{g}%
)\subset\mathcal{C}\ell_{\mathrm{Spin}_{1,3}^{e}}^{r}(M,\mathtt{g})$, we have%

\begin{equation}
\mathbf{\Psi e^{l}}=\mathbf{\Psi},~~~~~~\mathbf{e^{r} \Phi}=\mathbf{\Phi}.
\end{equation}

A section $\mathbf{\Psi}$ of $I^{l}(M,\mathtt{g})$ \emph{(}respectively
$\mathbf{\Phi}$ of $I^{r}(M,\mathtt{g})$\emph{)} is called a left
\emph{(}respectively right\emph{)} ideal algebraic spinor field.
\end{definition}

\begin{definition}
A Dirac-Hestenes spinor field \emph{(DHSF)} associated with $\mathbf{\Psi}$ is
a section\footnote{$\mathcal{C}\ell_{\mathrm{Spin}_{1,3}^{e}}^{0l}
(M,\mathtt{g})$ (respectively $\mathcal{C}\ell_{\mathrm{Spin}_{1,3}^{e}}%
^{0r}(M,\mathtt{g})$) denotes the even subbundle of $\mathcal{C}%
\ell_{\mathrm{Spin}_{1,3}^{e}}^{l}(M,\mathtt{g})$ (respectively $\mathcal{C}%
\ell_{\mathrm{Spin}_{1,3}^{e}}^{r}(M,\mathtt{g})$).} $\varPsi$ of
$\mathcal{C}\ell_{\mathrm{Spin}_{1,3}^{e}}^{0l}(M,\mathtt{g})\subset
\mathcal{C}\ell_{\mathrm{Spin}_{1,3}^{e}}^{l}(M,\mathtt{g})$ \emph{(}%
respectively a section $\varPhi$ of \break$\mathcal{C}\ell_{\mathrm{Spin}%
_{1,3}^{e}}^{0r}(M,\mathtt{g})\subset\mathcal{C}\ell_{\mathrm{Spin}_{1,3}^{e}%
}^{r}(M,\mathtt{g})$\emph{)} such that\footnote{For any $\mathbf{\Psi}$ the
DHSF always exist, see \cite{rc2007}.}
\begin{equation}
\mathbf{\Psi}=\varPsi\mathbf{e^{l}},~~~~~~~\mathbf{\Phi}=\mathbf{e^{r}%
}\varPhi. \label{DHSF}%
\end{equation}

\end{definition}

\begin{definition}
There are natural pairings:%
\begin{align}
\sec\mathcal{C}\ell_{\mathrm{Spin}_{1,3}^{e}}^{l}(M,\mathtt{g})\times
\sec\mathcal{C}\ell_{\mathrm{Spin}_{1,3}^{e}}^{r}(M,\mathtt{g})  &
\rightarrow\sec\mathcal{C}\ell(M,\mathtt{g}),\label{PAIRING}\\
\sec\mathcal{C}\ell_{\mathrm{Spin}_{1,3}^{e}}^{r}(M,\mathtt{g})\times
\sec\mathcal{C}\ell_{\mathrm{Spin}_{1,3}^{e}}^{l}(M,\mathtt{g})  &
\rightarrow\mathcal{F(}M,\mathbb{R}_{1,3}),
\end{align}
such that given a section $\alpha$ of $\mathcal{C}\ell_{\mathrm{Spin}%
_{1,3}^{e}}^{l}(M,\mathtt{g})$ and a section $\beta$ of $\mathcal{C}%
\ell_{\mathrm{Spin}_{1,3}^{e}}^{r}(M,\mathtt{g})$ and selecting
representatives $(p,a)$ for $\alpha(x)$ and $(p,b)$ for $\beta(x)$
\emph{(}$p\in\pi^{-1}\left(  x\right)  $) it is%
\begin{align}
(\alpha\beta)  &  :=[(p;ab)]\in\mathcal{C}\ell(M,\mathtt{g}),\\
(\beta\alpha)(x)  &  :=ba\in\mathbb{R}_{1,3}.
\end{align}
If alternative representatives $(pu^{-1},ua)$ and $(pu^{-1},bu^{-1})$ are
chosen for $\alpha(x)$ and $\beta(x)$ we have $[(pu^{-1};uabu^{-1})]$, that,
by \ref{eq_cliff}, represents the same element on $\mathcal{C}\ell
(M,\mathtt{g})$, and $(bu^{-1}ua)=ba$; thus $(\alpha\beta)(x)$ and
$(\beta\alpha)(x)$ are a well defined. Following the same procedure we could
define the actions \cite{rc2007}:%
\begin{align}
\sec\mathcal{C}\ell(M,\mathtt{g})\times\sec\mathcal{C}\ell_{\mathrm{Spin}%
_{1,3}^{e}}^{l}(M,\mathtt{g})  &  \rightarrow\sec\mathcal{C}\ell
_{\mathrm{Spin}_{1,3}^{e}}^{l}(M,\mathtt{g}),\\
\sec\mathcal{C}\ell_{\mathrm{Spin}_{1,3}^{e}}^{r}(M,\mathtt{g})\times
\sec\mathcal{C}\ell(M,\mathtt{g})  &  \rightarrow\sec\mathcal{C}%
\ell_{\mathrm{Spin}_{1,3}^{e}}^{r}(M,\mathtt{g}),\\
\sec\mathcal{C}\ell_{\mathrm{Spin}_{1,3}^{e}}^{l}(M,\mathtt{g})\times
\mathbb{R}_{1,3}  &  \rightarrow\sec\mathcal{C}\ell_{\mathrm{Spin}_{1,3}^{e}%
}^{l}(M,\mathtt{g}),\\
\mathbb{R}_{1,3}\times\sec\mathcal{C}\ell_{\mathrm{Spin}_{1,3}^{e}}%
^{r}(M,\mathtt{g})  &  \rightarrow\sec\mathcal{C}\ell_{\mathrm{Spin}_{1,3}%
^{e}}^{r}(M,\mathtt{g}).
\end{align}

\end{definition}

Given a local trivialization of $\mathcal{C}\ell(M,\mathtt{g})$ (or
$\mathcal{C}\ell_{\mathrm{Spin}_{1,3}^{e}}^{l}(M,\mathtt{g})$, $\mathcal{C}%
\ell_{\mathrm{Spin}_{1,3}^{e}}^{r}(M,\mathtt{g})$) ($\mathcal{U}\subset M$)%
\begin{equation}
\phi_{U}:\pi^{-1}(\mathcal{U})\rightarrow\mathcal{U}\times\mathbb{R}_{1,3},
\end{equation}
we can define a local unit section by $\mathbf{1}_{U}(x)=\phi_{U}^{-1}(x,1)$.
For $\mathcal{C}\ell(M,\mathtt{g})$, it is easy to show that a global unit
section always exist, independently of the fact that $M$ \ is parallizable or
not. For the bundles $\mathcal{C}\ell_{\mathrm{Spin}_{p,q}^{e}}^{l}%
(M,\mathtt{g})$, $\mathcal{C}\ell_{\mathrm{Spin}_{p,q}^{e}}^{r}(M,\mathtt{g}%
)$) ($\dim$ $M=p+q$) there exist a global unit sections if, and only if,
$P_{\mathrm{Spin}_{p,q}^{e}}(M,\mathtt{g})$ is trivial \cite{r2004,rc2007}. In
our case we know, by Geroch theorem, that $M$ is parallelizable and we can
define global unit sections on $\mathcal{C}\ell_{\mathrm{Spin}_{1,3}^{e}}%
^{l}(M,\mathtt{g})$ and $\mathcal{C}\ell_{\mathrm{Spin}_{1,3}^{e}}%
^{r}(M,\mathtt{g})$.

Let $\boldsymbol{\mathbf{\Xi}}$$_{u}$ be a section of $P_{\mathrm{Spin}%
_{1,3}^{e}}(M,\mathtt{g})$, i.e., a spin frame. We recall, in order to fix
notations, that sections of $\mathcal{C}\ell(M,\mathtt{g})$ $I^{l}%
(M,\mathtt{g})$, $I^{r}(M,\mathtt{g})$ $\mathcal{C}\ell_{\mathrm{Spin}%
_{1,3}^{e}}^{l}(M,\mathtt{g})$, $\mathcal{C}\ell_{\mathrm{Spin}_{1,3}^{e}}%
^{r}(M,\mathtt{g})$ are, respectively, the equivalence classes%
\begin{gather}
\boldsymbol{C}=[(\boldsymbol{\mathbf{\Xi}}_{u},\mathcal{C}%
_{\boldsymbol{\mathbf{\Xi}}_{u}})],\nonumber\\
\mathbf{\Psi}=[(\boldsymbol{\mathbf{\Xi}}_{u},\boldsymbol{\Psi}%
_{\boldsymbol{\mathbf{\Xi}}_{u}})]~~~\mathbf{\Phi}=[(\boldsymbol{\mathbf{\Xi}%
}_{u},\boldsymbol{\Phi}_{\boldsymbol{\mathbf{\Xi}}_{u}}%
)],~~~\varPsi =[(\boldsymbol{\mathbf{\Xi}}_{u}%
,\varPsi_{\boldsymbol{\mathbf{\Xi}}_{u}}%
)],~~~\varPhi=[(\boldsymbol{\mathbf{\Xi}}_{u},\varPhi_{\boldsymbol{\mathbf{\Xi
}}_{u}})]. \label{notation}%
\end{gather}

\begin{remark}
When convenient, we will write $\mathcal{C}_{\Xi_{u}}\in\sec\mathcal{C}%
\ell(M,\mathtt{g})$ to mean that there exists a section $\boldsymbol{C}$ of
the Clifford bundle $\mathcal{C}\ell(M,\mathtt{g})$ defined by
$[(\boldsymbol{\mathbf{\Xi}}_{u},\mathcal{C}_{\boldsymbol{\mathbf{\Xi}}_{u}%
})]$. Analogous notations will be used for sections of the other bundles
introduced above. Also, when there is no chance of confusion on the chosen
spinor frame, we will write $\mathcal{C}_{\boldsymbol{\mathbf{\Xi}}_{u}}$
simply as $\mathcal{C}$.
\end{remark}

For each spin frame, say $\boldsymbol{\mathbf{\Xi}}$$_{0}$, let $\mathbf{1}%
_{\boldsymbol{\mathbf{\Xi}}_{0}}^{l}$ and $\mathbf{1}_{\boldsymbol{\mathbf{\Xi
}}_{0}}^{r}$ be the global unit sections of $\mathcal{C\ell}_{\mathrm{Spin}%
_{1,3}^{e}}^{l}(M,\mathtt{g})$ and $\mathcal{C\ell}_{\mathrm{Spin}_{1,3}^{e}%
}^{r}(M,\mathtt{g})$, given by%

\begin{equation}
\mathbf{1}_{\boldsymbol{\mathbf{\Xi}}_{0}}^{r}:=[(\Xi_{0},1)],~~~~~\mathbf{1}%
_{\boldsymbol{\mathbf{\Xi}}_{0}}^{l}:=[(\Xi_{0},1)]. \label{repre3}%
\end{equation}

\begin{remark}
Before proceeding note that given another spin frame $\Xi_{u}=\Xi_{0}u$, where
$u:M\rightarrow\mathrm{Spin}_{1,3}^{e}\subset\mathbb{R}_{1,3}^{0}%
\subset\mathbb{R}_{1,3}$ we define the sections $\mathbf{1}%
_{\boldsymbol{\mathbf{\Xi}}_{u}}^{r}$ of $\mathcal{C\ell}_{\mathrm{Spin}%
_{1,3}^{e}}^{r}(M,\mathtt{g})$ and $\mathbf{1}_{\boldsymbol{\mathbf{\Xi}}_{0}%
}^{l}$ of $\mathcal{C\ell}_{\mathrm{Spin}_{1,3}^{e}}^{l}(M,\mathtt{g})$ by
\begin{equation}
\mathbf{1}_{\boldsymbol{\mathbf{\Xi}}_{u}}^{r}:=[(\Xi_{u},1)],~~~~~\mathbf{1}%
_{\boldsymbol{\mathbf{\Xi}}_{u}}^{l}:=[(\Xi_{u},1)]. \label{rep4}%
\end{equation}

It has been proved in \emph{\cite{r2004,rc2007}} that the relation between
$\mathbf{1}_{\boldsymbol{\mathbf{\Xi}}}^{r}$ and $\mathbf{1}%
_{\boldsymbol{\mathbf{\Xi}}_{0}}^{r}$ and between $\mathbf{1}%
_{\boldsymbol{\mathbf{\Xi}}}^{l}$ and $\mathbf{1}_{\boldsymbol{\mathbf{\Xi}%
}_{0}}^{l}$\ are given by
\begin{equation}
\mathbf{1}_{\boldsymbol{\mathbf{\Xi}}_{u}}^{r}=u^{-1}\mathbf{1}%
_{\boldsymbol{\mathbf{\Xi}}_{0}}^{r}=\mathbf{1}_{\boldsymbol{\mathbf{\Xi}}%
_{0}}^{r}U^{-1},~~~~~\mathbf{1}_{\boldsymbol{\mathbf{\Xi}}_{u}}^{l}%
=U\mathbf{1}_{\boldsymbol{\mathbf{\Xi}}_{0}}^{l}=\mathbf{1}%
_{\boldsymbol{\mathbf{\Xi}}_{0}}^{l}u \label{rep5}%
\end{equation}
where $U$ is the section of $\mathcal{C}\ell(M,\mathtt{g})$ defined by the
equivalence class
\begin{equation}
U=[(\Xi_{0},u)]. \label{rep6}%
\end{equation}

\end{remark}

The unity sections $\mathbf{1}_{\boldsymbol{\mathbf{\Xi}}_{u}}^{l}$ and
$\mathbf{1}_{\boldsymbol{\mathbf{\Xi}}_{u}}^{r}$satisfies the important
relations\footnote{$\mathcal{C}\ell^{0}(M,\mathtt{g})$ denotes the even
subbundle of $\mathcal{C}\ell(M,\mathtt{g})$.}%

\begin{equation}
\mathbf{1}_{\boldsymbol{\mathbf{\Xi}}_{u}}^{l}\mathbf{1}%
_{\boldsymbol{\mathbf{\Xi}}_{u}}^{r}=1\in\sec\mathcal{C}\ell(M,\mathtt{g}%
),~\mathbf{1}_{\boldsymbol{\mathbf{\Xi}}_{u}}^{r}\mathbf{1}%
_{\boldsymbol{\mathbf{\Xi}}_{u}}^{l}=1\in\mathcal{F(}M,\mathbb{R}_{1,3}),
\label{rep6a}%
\end{equation}

\begin{definition}
A representative of a DHSF $\varPsi$ \emph{(}respectively $\varPhi$\emph{)} in
the Clifford bundle $\mathcal{C}\ell(M,\mathtt{g})$ relative to a spin frame
$\mathbf{\Xi}_{u}$ is a section $\boldsymbol{\psi}_{\mathbf{\Xi}_{u}%
}=[(\mathbf{\Xi}_{u},\psi_{\mathbf{\Xi}_{u}})]$ of $\mathcal{C}\ell
^{0}(M,\mathtt{g})$ \emph{(}respectively, a section\emph{ }$\boldsymbol{\phi
}_{\mathbf{\Xi}_{u}}=[(\mathbf{\Xi}_{u},\phi_{\mathbf{\Xi}_{u}})]$ of
$\mathcal{C}\ell^{0}(M,\mathtt{g})$\emph{) }given by \emph{\cite{r2004,rc2007}%
}
\end{definition}

\
\begin{equation}
\boldsymbol{\psi}_{\mathbf{\Xi}_{u}}=\Psi\mathbf{1}_{\mathbf{\Xi}_{u}}%
^{r},~~~~\mathbf{1}_{\mathbf{\Xi}_{u}}^{l}\Phi=\boldsymbol{\phi}_{\mathbf{\Xi
}_{u}}. \label{RDHSF}%
\end{equation}
Representatives in the Clifford bundle of $\varPsi$ relative to spin frames,
say $\mathbf{\Xi}_{u^{\prime}}$ and $\mathbf{\Xi}_{u}$, are related
by\footnote{This relation has been used in \cite{mr02014} to define a DHSF as
an appropriate equivalence class of even sections of the Clifford bundle
$\mathcal{C\ell}(M,\mathtt{g})$.}%

\begin{equation}
\boldsymbol{\psi}_{\mathbf{\Xi}_{u^{\prime}}}U^{\prime-1}=\boldsymbol{\psi
}_{\mathbf{\Xi}_{u}}U^{-1}. \label{30}%
\end{equation}

Also, representatives in the Clifford bundle of\ $\varPhi$ relative to spin
frames $\mathbf{\Xi}_{u^{\prime}}$ and $\mathbf{\Xi}_{u}$ are related by%

\begin{equation}
U^{\prime}\boldsymbol{\phi}_{\mathbf{\Xi}_{u^{\prime}}}=U\boldsymbol{\phi
}_{\mathbf{\Xi}_{u}}.
\end{equation}

\section{On the Concept of Lie Derivatives of Spinor Fields in Lorentzian
Manifolds}

Lie derivatives of tensor fields are defined once we give the concept of the
push forward and pullback mappings (which serves the purpose of defining the
image of the tensor field) associated to one-parameter groups of
diffeomorphisms generated by vector fields. These concepts are well known and
very important in the derivation of conserved currents in physical theories.

It happens that physical theories need also the concept of spinor fields
living on a Lorentzian manifold and the question arises as how to define a
meaningful image for these objects under a diffeomorphism. There are a lot of
different approaches to the subject, as the reader can learn consulting, e.g.,
\cite{bt1987,bm,bg1992,cru,dabrowski,F1,F,gm1,gm2,gursey,hv2000,kosmann0,kosmann1,kosmann2,kosmann3,lichnerowicz,pr2,weinberg}%
. In what follows using the definition of left (and right) \emph{real }spinor
fields (in particular, Dirac-Hestenes spinor fields)
\cite{r2004,mr2004,rc2007} living in a Lorentzian manifold and their
representatives in the Clifford bundle we give a geometrical motivated
definition for their images and a corresponding definition of the Lie
derivative for spinor and Clifford fields. We compare our definition with some
others appearing in the literature.

We already recalled that fixing a global spinor\ basis\footnote{Such a basis
must exists according to Geroch Theorem (\ref{gerochh}) \cite{geroch}.}
\textbf{$\Xi$}$_{0}(x)=(x,u_{0}(x))$ for $P_{\mathrm{Spin}_{1,3}^{e}%
}(M,\boldsymbol{g}$), and given an algebraic spinor $\mathbf{\Psi}$, the
associated DHSF $\varPsi$ can be represented in the Clifford bundle by the
object\footnote{The notation $\boldsymbol{P}\in\sec%
%TCIMACRO{\tbigwedge \nolimits^{p}}%
%BeginExpansion
{\textstyle\bigwedge\nolimits^{p}}
%EndExpansion
T^{\ast}M\hookrightarrow\sec\mathcal{C}\ell(M,\mathtt{g})$ means that there
exits a section $s_{P}$ of $%
%TCIMACRO{\tbigwedge \nolimits^{p}}%
%BeginExpansion
{\textstyle\bigwedge\nolimits^{p}}
%EndExpansion
T^{\ast}M\hookrightarrow\mathcal{C}\ell(M,\mathtt{g})$ given by or any $x\in
M$ by the equivalence class $[(\Xi_{0}(x),\boldsymbol{P}(x))]$ where
$\boldsymbol{P}:M\rightarrow%
%TCIMACRO{\tbigwedge \nolimits^{p}}%
%BeginExpansion
{\textstyle\bigwedge\nolimits^{p}}
%EndExpansion
\mathbb{R}^{1,3}\hookrightarrow\mathbb{R}_{1,3}$.}
\begin{equation}
\psi_{\mathbf{\Xi}_{0}}\in\sec\mathcal{C}\ell^{0}(M,\mathtt{g}). \label{lie1}%
\end{equation}

\begin{remark}
When $\psi_{\mathbf{\Xi}_{0}}\tilde{\psi}_{\mathbf{\Xi}_{0}}\neq0$ we can
easily show that $\psi_{\mathbf{\Xi}_{0}}$ has the following
representation\footnote{The product $\rho(x)^{\frac{1}{2}}e^{-\frac
{\boldsymbol{\tau}_{\boldsymbol{g}}(x)\beta(x)}{2}}R(x)$ in Eq.(\ref{lie2})
must be understood as meaning the product $[(\Xi_{0}(x),\rho(x)][(\Xi
_{0}(x),e^{-\frac{\boldsymbol{\tau}_{\boldsymbol{g}}(x)\beta(x)}{2}}%
)][(\Xi_{0}(x),R(x))]=[(\Xi_{0},\rho^{\frac{1}{2}}e^{-\frac{\boldsymbol{\tau
}_{\boldsymbol{g}}\beta}{2}}R)],$where $\rho(x),\beta(x)\in%
%TCIMACRO{\tbigwedge \nolimits^{0}}%
%BeginExpansion
{\textstyle\bigwedge\nolimits^{0}}
%EndExpansion
\mathbb{R}^{1,3}\hookrightarrow\mathbb{R}_{1,3}$, $R(x)\in\mathrm{Spin}%
_{1,3}^{e}\subset\mathbb{R}_{1,3}^{0}$ and $\boldsymbol{\tau}_{\boldsymbol{g}%
}(x)\in%
%TCIMACRO{\tbigwedge \nolimits^{4}}%
%BeginExpansion
{\textstyle\bigwedge\nolimits^{4}}
%EndExpansion
\mathbb{R}^{1,3}\hookrightarrow\mathbb{R}_{1,3}$.}%
\begin{equation}
\psi_{\mathbf{\Xi}_{0}}=\rho^{\frac{1}{2}}e^{-\frac{\boldsymbol{\tau
}_{\boldsymbol{g}}\beta}{2}}R, \label{lie2}%
\end{equation}
where $\rho,\beta\in\sec%
%TCIMACRO{\tbigwedge \nolimits^{0}}%
%BeginExpansion
{\textstyle\bigwedge\nolimits^{0}}
%EndExpansion
T^{\ast}M\hookrightarrow\sec\mathcal{C}\ell^{0}(M,\mathtt{g})$ and
\emph{\cite{lounesto}}
\begin{equation}
R=e^{\mathcal{F}}\in\sec\mathrm{Spin}_{1,3}^{e}(M,\mathtt{g})\hookrightarrow
\sec\mathcal{C}\ell^{0}(M,\mathtt{g}), \label{lie2a}%
\end{equation}
with $\mathcal{F}\in\sec%
%TCIMACRO{\tbigwedge \nolimits^{2}}%
%BeginExpansion
{\textstyle\bigwedge\nolimits^{2}}
%EndExpansion
T^{\ast}M\hookrightarrow\sec\mathcal{C}\ell^{0}(M,\mathtt{g}).$
\end{remark}

Let $\boldsymbol{\xi}\in\sec TM$ be a smooth vector field. For any $x\in M$
there exists an unique integral curve of $\boldsymbol{\xi}$, given by
$t\mapsto\mathrm{h}(t,x)$, with $x=\mathrm{h}(0,x)$. We recall that for
$(t,x)\in$ $I(x)\mathbb{\times M}$ ($I(x)\subset\mathbb{R}$) the mapping
\textrm{h}$:(t,x)\mapsto\mathrm{h}(t,x)$ is called the flow of
$\boldsymbol{\xi}$. We suppose in what follows that the mappings
\textrm{h}$_{t}:=\mathrm{h}(t,~):M\rightarrow M$, $x\mapsto x^{\prime
}=\mathrm{h}_{t}(x)$ generate a one-parameter group of diffeomorphisms of $M$
(i.e., $I(x)=\mathbb{R)}$ We have (for a fixed $x\in M$)
\begin{equation}
\boldsymbol{\xi}(\mathrm{h}(t,x))=\left.  \frac{d}{dt}\mathrm{h}%
(t,x)\right\vert _{t=0.}. \label{opg}%
\end{equation}

Now, recall that $\psi_{\mathbf{\Xi}_{0}}$ determines global $1$-form fields
on $M$, namely $\boldsymbol{V}^{\alpha}\in\sec%
%TCIMACRO{\tbigwedge \nolimits^{1}}%
%BeginExpansion
{\textstyle\bigwedge\nolimits^{1}}
%EndExpansion
T^{\ast}M\hookrightarrow\sec\mathcal{C}\ell(M,\mathtt{g})$ such that at
$x^{\prime}=\mathrm{h}_{t}(x)$
\begin{align}
\boldsymbol{V}^{\alpha}(x^{\prime})  &  =\psi_{\mathbf{\Xi}_{0}}(x^{\prime
})\boldsymbol{\gamma}^{\alpha}(x^{\prime})\tilde{\psi}_{\mathbf{\Xi}_{0}%
}(x^{\prime})=\rho(x^{\prime})R(x^{\prime})\boldsymbol{\gamma}^{a}(x^{\prime
})\tilde{R}(x^{\prime})\nonumber\\
&  =\rho(x^{\prime})L_{\beta}^{\alpha}(x^{\prime})\boldsymbol{\gamma}^{\beta
}(x^{\prime})=\rho(x^{\prime})\boldsymbol{\Gamma}^{\alpha}(x^{\prime}).
\label{elie3}%
\end{align}

The \emph{pullbacks} of $\boldsymbol{\gamma}^{\alpha}$, $\boldsymbol{\Gamma
}^{\alpha}$ and $\boldsymbol{V}^{\alpha}$ are the fields
\begin{equation}
\boldsymbol{\gamma}_{t}^{\prime\alpha}=\mathrm{h}_{t}^{\ast}\boldsymbol{\gamma
}^{\alpha},~~~\boldsymbol{\Gamma}_{t}^{\prime\alpha}=\mathrm{h}_{t}^{\ast
}\boldsymbol{\Gamma}^{\alpha},~~~~\boldsymbol{V}_{t}^{\prime\alpha}%
=\mathrm{h}_{t}^{\ast}\boldsymbol{V}^{\alpha}.~~~ \label{lie3a}%
\end{equation}

At $x\in$ $M$, expressing the diffeomorphism $x^{\prime}=\mathrm{h}_{t}(x)$ in
a local coordinate chart covering \ $\mathcal{U\subset}M$ it is
\begin{equation}
\boldsymbol{\gamma}_{t}^{\prime\alpha}(x)=h_{t\mu}^{\alpha}(x^{\prime
}(x))\frac{\partial x^{\prime\mu}}{\partial x^{\nu}}dx^{\nu},
\end{equation}
with similar formulas for the $\boldsymbol{\Gamma}_{t}^{\prime\alpha}$ and
$\boldsymbol{V}_{t}^{\prime\alpha}$. In particular take notice that%
\begin{align}
\boldsymbol{V}_{t}^{\prime\alpha}(x)  &  =\rho(x^{\prime}(x))L_{\beta}%
^{\alpha}(x^{\prime}(x))\boldsymbol{\gamma}^{\prime\beta}(x)\nonumber\\
&  =\rho(x^{\prime}(x))R^{\prime}(x)\boldsymbol{\gamma}^{\prime\alpha
}(x)\tilde{R}^{\prime}(x). \label{lie3b}%
\end{align}

\begin{remark}
It is clear that $\{\boldsymbol{\gamma}_{t}^{\prime\alpha}%
(x)\},\{\boldsymbol{\Gamma}_{t}^{\prime\alpha}(x))\}$ are not orthonormal
basis for $T_{x}M$ relative to $\mathrm{g}_{x}$ unless $\boldsymbol{\xi}$ is a
Killing vector field, but of course, they are orthonormal basis of $T_{x}M$
relative to $\mathrm{g}_{x}^{\prime}=\left.  \mathrm{h}_{t}^{\ast}%
\mathrm{g}\right\vert _{x}$.
\end{remark}

Eq.(\ref{lie2a}) says that $R\in\sec\mathrm{Spin}_{1,3}^{e}\subset
\sec\mathcal{C}\ell^{0}(M,\mathtt{g})$ is the exponential of a biform field
\cite{lounesto}, say $R=e^{\mathcal{F}(x)}$ Thus, we see that there exists
\emph{no} difficulty in defining the pullback\footnote{Or of more generally,
even for a\ non invertible $\psi_{\mathbf{\Xi}_{0}}\in\mathcal{C}\ell
^{0}(M,\mathtt{g})$, which is a sum of even nonhomogeneous differetial forms.}
of $\rho^{\frac{1}{2}}e^{-\frac{\boldsymbol{\tau}_{\boldsymbol{g}}\beta}{2}%
}e^{\mathcal{F}(x)}$ under $\mathrm{h}_{t}$ (or of more generally, for any
$\psi_{\mathbf{\Xi}_{0}}\in\mathcal{C}\ell^{0}(M,\mathtt{g})$), which will be
written as%
\begin{equation}
\rho^{\frac{1}{2}}(x^{\prime}(x))e^{-\frac{\tau_{\boldsymbol{gt}}^{\prime
}(x)\beta(x^{\prime}(x))}{2}}e^{\mathcal{F}_{t}^{\prime}(x)}. \label{lie3bb}%
\end{equation}
However, we immediately have a

\begin{quotation}
\noindent\textbf{Problem}: The object defined by Eq.(\ref{lie3bb}) is of
course, a representative in $\mathcal{C}\ell^{0}(M,g)$ of some Dirac-Hestenes
spinor field but there is no way to know to which the spinor frame that object
is associated.
\end{quotation}

Thus, we must find another way to define the Lie derivative for spinor fields.
Our way, as we will see, is based in a geometric motivated definition for the
concept of image of Clifford and spinor fields under diffeomorphisms generated
by one-parameter group associated to an arbitrary vector field
$\boldsymbol{\xi}$. But, we need first to introduce some results, starting
with the

\begin{proposition}
\label{liekilling}Let $\pounds _{\boldsymbol{\xi}}$ denotes the standard Lie
derivative of tensor fields. If $\boldsymbol{\xi}$ is a Killing vector field
then
\begin{align}
\pounds _{\boldsymbol{\xi}}\boldsymbol{\gamma}^{\alpha} &  =\frac{1}{4}%
[L(\xi)+d\xi,\boldsymbol{\gamma}^{\alpha}]\label{lie3b12}\\
&  =\boldsymbol{D}_{\xi}\boldsymbol{\gamma}^{\alpha}+\frac{1}{4}%
[d\xi,\boldsymbol{\gamma}^{\alpha}]\label{lie3b12a}%
\end{align}
with
\begin{equation}
L(\xi):=\frac{1}{2}(c_{\alpha\kappa\iota}+c_{\kappa\alpha\iota}+c_{\iota
\alpha\kappa})\xi^{\kappa}\boldsymbol{\gamma}^{\alpha}\wedge\boldsymbol{\gamma
}^{\iota}\label{lie3b12b}%
\end{equation}
where $c_{\cdot\kappa\iota}^{\alpha\cdot\cdot}$ are the structure coefficients
of the basis $\{\boldsymbol{e}_{\alpha}\}$ dual of $\{\boldsymbol{\gamma
}^{\alpha}\}$.
\end{proposition}

\begin{proof}
First recall that Eq.(\ref{lie3b12a})\ is clearly equal to Eq.(\ref{lie3b12})
since once $\boldsymbol{D}_{\boldsymbol{e}_{\kappa}}\boldsymbol{\gamma
}^{\alpha}=-\omega_{\cdot\kappa\iota}^{\alpha\cdot\cdot}\boldsymbol{\gamma
}^{\iota}$ it is%
\begin{equation}
\omega_{\alpha\kappa\iota}=\frac{1}{2}(c_{\alpha\kappa\iota}+c_{\kappa
\alpha\iota}+c_{\iota\alpha\kappa}). \label{lie3b12c}%
\end{equation}
It follows that%

\begin{equation}
L(\xi)=2\omega_{\xi} \label{lie3b12d}%
\end{equation}
where%
\begin{equation}
\omega_{\xi}=\frac{1}{2}\xi^{\kappa}\omega_{\alpha\kappa\iota}%
\boldsymbol{\gamma}^{\alpha}\wedge\boldsymbol{\gamma}^{\iota} \label{lie3b12e}%
\end{equation}
is the \textquotedblleft connection biform\textquotedblright\ and so
\begin{equation}
\frac{1}{4}[L(\xi),\boldsymbol{\gamma}^{\alpha}]=\frac{1}{2}[\omega_{\xi
},\boldsymbol{\gamma}^{\alpha}]=-\boldsymbol{\gamma}^{\alpha}\lrcorner
\omega_{\xi}=\boldsymbol{D}_{\mathbf{\xi}}\boldsymbol{\gamma}^{\alpha}.
\label{lie3b12f}%
\end{equation}
Now, recalling Cartan's magical formula, and the following identities
\begin{equation}
\pounds _{\boldsymbol{\xi}}\boldsymbol{\gamma}^{\alpha}=\xi\lrcorner
d\boldsymbol{\gamma}^{\alpha}+d(\xi\lrcorner\boldsymbol{\gamma}^{\alpha}%
)=\xi\lrcorner d\boldsymbol{\gamma}^{\alpha}+d(\boldsymbol{\gamma}^{\alpha
}\lrcorner\xi) \label{lie3b2}%
\end{equation}%
\begin{equation}
d=\boldsymbol{\gamma}^{\iota}\wedge\boldsymbol{D}_{\boldsymbol{e}_{\iota}},
\end{equation}
we have%
\begin{align*}
\pounds _{\boldsymbol{\xi}}\boldsymbol{\gamma}^{\alpha}  &  =\xi
\lrcorner(\boldsymbol{\gamma}^{\iota}\wedge\boldsymbol{D}_{\boldsymbol{e}%
_{\iota}}\boldsymbol{\gamma}^{\alpha})+d(\boldsymbol{\xi}^{\alpha})\\
&  =(\xi\lrcorner\boldsymbol{\gamma}^{\iota})\boldsymbol{D}_{\boldsymbol{e}%
_{\iota}}\boldsymbol{\gamma}^{\alpha}-(\xi\lrcorner\boldsymbol{D}%
_{\boldsymbol{e}_{\iota}}\boldsymbol{\gamma}^{\alpha})\boldsymbol{\gamma
}^{\iota}+\boldsymbol{e}_{\iota}(\boldsymbol{\xi}^{\alpha})\boldsymbol{\gamma
}^{\iota}\\
&  =\boldsymbol{\xi}^{\iota}\boldsymbol{D}_{\boldsymbol{e}_{\iota}%
}\boldsymbol{\gamma}^{\alpha}+(\xi^{\mathbf{\kappa}}\boldsymbol{\gamma
}_{\mathbf{\kappa}}\lrcorner\omega_{\mathbf{\cdot\iota\lambda}}%
^{\mathbf{\alpha}\cdot\cdot}\boldsymbol{\gamma}^{\lambda})\boldsymbol{\gamma
}^{\iota}+\boldsymbol{e}_{\iota}(\boldsymbol{\xi}^{\alpha})\boldsymbol{\gamma
}^{\iota}\\
&  =\boldsymbol{D}_{\mathbf{\xi}}\boldsymbol{\gamma}^{\alpha}+[\xi
^{\mathbf{\lambda}}\omega_{\mathbf{\cdot\iota\lambda}}^{\mathbf{\alpha}%
\cdot\cdot}]\boldsymbol{\gamma}^{\iota}+\boldsymbol{e}_{\iota}(\boldsymbol{\xi
}^{\alpha})\boldsymbol{\gamma}^{\iota}\\
&  =\boldsymbol{D}_{\mathbf{\xi}}\boldsymbol{\gamma}^{\alpha}+(\boldsymbol{D}%
_{\iota}\xi^{\alpha})\boldsymbol{\gamma}^{\iota}%
\end{align*}
We get,
\begin{equation}
\pounds _{\boldsymbol{\xi}}\boldsymbol{\gamma}^{\alpha}=\boldsymbol{D}%
_{\boldsymbol{\xi}}\boldsymbol{\gamma}^{\alpha}+(\boldsymbol{D}_{\iota}%
\xi^{\alpha})\boldsymbol{\gamma}^{\iota}. \label{lieb4a}%
\end{equation}

Now, it remains to show that for $\xi$ a Killing vector field $\frac{1}%
{4}[d\xi,\boldsymbol{\gamma}^{\alpha}]$ is equal to $(\boldsymbol{D}_{\iota
}\xi^{\alpha})\boldsymbol{\gamma}^{\iota}$. We have
\begin{gather}
\frac{1}{4}[d\xi,\boldsymbol{\gamma}^{\alpha}]=-\frac{1}{2}\boldsymbol{\gamma
}^{\alpha}\lrcorner d\xi\nonumber\\
=-\frac{1}{2}\boldsymbol{\gamma}^{\alpha}\lrcorner(\boldsymbol{\gamma}^{\iota
}\wedge\boldsymbol{D}_{\boldsymbol{e}_{\iota}}\xi)=-\frac{1}{2}%
\boldsymbol{\gamma}^{\alpha}\lrcorner\{\boldsymbol{\gamma}^{\iota}%
\wedge(\boldsymbol{D}_{\iota}\xi_{\kappa})\boldsymbol{\gamma}^{\kappa
}\}\nonumber\\
=\frac{1}{2}\boldsymbol{D}_{\iota}\xi_{\kappa}\boldsymbol{\gamma}^{\alpha
}\lrcorner\{\boldsymbol{\gamma}^{\kappa}\wedge\boldsymbol{\gamma}^{^{\iota}%
}\}=\boldsymbol{D}_{\iota}\xi^{\alpha}\boldsymbol{\gamma}^{\iota}-\frac{1}%
{2}\left(  \boldsymbol{D}_{\iota}\xi_{\kappa}+\boldsymbol{D}_{\kappa}\xi
_{l}\right)  \eta^{\alpha\iota}\boldsymbol{\gamma}^{\kappa}=\boldsymbol{D}%
_{\iota}\xi^{\alpha}\boldsymbol{\gamma}^{\iota}\label{lieb6}%
\end{gather}
and the proposition is proved.
\end{proof}

\begin{remark}
We emphasize that in the derivation of the previous result it has not been
used the fact that $\pounds _{\boldsymbol{\xi}}$ must be a derivation in the
Clifford bundle, since all operations done requires only simple and well known
formulas from the calculus of differential forms.
\end{remark}

\begin{remark}
Moreover, one can easily show using the previous results that when
$\mathcal{C}\in\sec\mathcal{C\ell(}M,\mathtt{g})$ and $\boldsymbol{\xi}\in\sec
TM$ is a Killing vector field then
\begin{equation}
\pounds _{\boldsymbol{\xi}}\mathcal{C}=\mathfrak{d}_{\xi}\mathcal{C+}\frac
{1}{4}[\mathbf{S}(\xi),\mathcal{C}]. \label{lieb16a}%
\end{equation}

\end{remark}

Indeed, Eq.(\ref{lieb16a}) follows trivially by induction and noting that
$\pounds _{\boldsymbol{\xi}}(\mathcal{AB})=\pounds _{\boldsymbol{\xi}%
}(\mathcal{A)B}+\mathcal{A}\pounds _{\boldsymbol{\xi}}(\mathcal{B})$, where
$\mathcal{A},\mathcal{B}\in\sec\mathcal{C\ell(}M,\mathtt{g}),$ when
$\boldsymbol{\xi}\in\sec TM$ is a Killing vector field.

This suggests that $L(\xi)$ should be involved in the definition of the Lie
derivative of spinor fields. Based on this, and recalling Eq.(\ref{lie2a}) we
propose that the \emph{spinor lifting} of an integral curve of a generic
smooth vector field $\boldsymbol{\xi}\in\sec TM$ to $P_{\mathrm{Spin}^{e}%
1,3}(M,\mathtt{g})$ in the parallelizabe manifold $M$ equipped with the global
orthonormal cobasis $\{\boldsymbol{\gamma}^{\alpha}\}$ is given by the
following\smallskip

\begin{definition}
\label{SPLIFT}Consider the integral curve $\mathrm{h}_{t}:\mathbb{R}%
\rightarrow M$ of an arbitrary smooth vector field $\boldsymbol{\xi}$. The
spinor lifiting $\breve{h}_{t}$ of $\mathrm{h}_{t}$ to $P_{\mathrm{Spin}%
_{1,3}^{e}}(M,\mathtt{g})$ is the curve
\begin{gather}
\breve{h}_{t}(p)=(h_{t}(\pi(p)),au_{t}(h_{t}(\pi(p)))\\
u_{t}(x):=e^{-\frac{1}{4}t(\mathbf{S}(\xi)(x))}\in\mathrm{Spin}_{1,3}^{e},\\
\mathbf{S}(\xi)=L(\xi)+d\xi,
\end{gather}
with $\pi(p)=\pi((x,a))=x$.
\end{definition}

To see why the above definition is really important consider that for $t<<1$
it is
\begin{equation}
u_{t}=1-\frac{1}{4}t\mathbf{S}(\xi)+O(t^{2})+\cdots\label{lieb11}%
\end{equation}
Then, we have for $t<<1$ that
\begin{align}
u_{t}^{-1}\boldsymbol{\gamma}^{\alpha}u_{t}  &  =\{1+\frac{1}{4}%
t\mathbf{S}(\xi)+O(t^{2})+\cdots\}\boldsymbol{\gamma}^{\alpha}\{1-\frac{1}%
{4}t\mathbf{S}(\xi)+O(t^{2})+\cdots\}\nonumber\\
&  =\boldsymbol{\gamma}^{\alpha}+\frac{1}{4}t[\mathbf{S}(\xi
),\boldsymbol{\gamma}^{\alpha}]+O(t^{2})+\cdots\label{lieb121}%
\end{align}
Deriving in $t=0$ we obtain the expression of the previous proposition.

Now, recall that the pullback $\boldsymbol{\gamma}_{t}^{\prime\alpha
}=\mathrm{h}_{t}^{\ast}\boldsymbol{\gamma}^{\alpha}$ when $\boldsymbol{\xi}$
\ is any (arbitrary) vector field for $t<<1$ is%

\begin{equation}
\boldsymbol{\gamma}_{t}^{\prime\alpha}(x)=\boldsymbol{\gamma}^{\alpha
}(x)+t\pounds _{\boldsymbol{\xi}}\boldsymbol{\gamma}^{\alpha}(x)+O(t^{2}%
)+\cdots\label{lie3b1}%
\end{equation}

Using the Proposition (\ref{liekilling}) (valid when $\xi$ is
Killing),\textbf{ }comparing Eq.(\ref{lie3b1}) with Eq.(\ref{lieb121}) and
recalling Eq.(\ref{lie3b12}), we see that up to the\emph{ first order} we have
for this case%
\begin{equation}
\boldsymbol{\gamma}_{t}^{\prime\alpha}(x)=u_{t}^{-1}(x)\boldsymbol{\gamma
}^{\alpha}(x)u_{t}(x)\label{lieb13}%
\end{equation}

\begin{remark}
We could write the first member of \emph{Eq.(\ref{lieb13})} and keeping terms
up to first order as%
\begin{equation}
\boldsymbol{\gamma}_{t}^{\prime\alpha}=\Lambda_{t\beta}^{\alpha}%
\boldsymbol{\gamma}^{\beta}=(\delta_{\beta}^{\alpha}+t\Sigma_{\beta}^{\alpha
})\boldsymbol{\gamma}^{\beta} \label{lieb14}%
\end{equation}
where $\Lambda_{t\beta}^{\alpha}(x)\in\mathrm{SO}_{1,3}^{e}$ for any
$x\in\mathcal{U}\subset M$. Thus \
\begin{equation}
\boldsymbol{\gamma}^{\beta}\Sigma_{\beta}^{\alpha}=\frac{1}{4}[\mathbf{S}%
(\xi),\boldsymbol{\gamma}^{\alpha}]=-\frac{1}{2}\boldsymbol{\gamma}^{\alpha
}\lrcorner(\mathbf{S}(\xi)) \label{lieb15}%
\end{equation}
and then
\begin{align}
\Sigma_{\kappa}^{\alpha}  &  =\boldsymbol{\gamma}_{\kappa}\lrcorner
\boldsymbol{\gamma}^{\beta}\Sigma_{t\beta}^{\alpha}=-\frac{1}{2}%
\boldsymbol{\gamma}_{\kappa}\lrcorner(\boldsymbol{\gamma}^{\alpha}%
\lrcorner\mathbf{S}(\xi))\nonumber\\
&  =-\frac{1}{2}(\boldsymbol{\gamma}_{\kappa}\wedge\boldsymbol{\gamma}%
^{\alpha})\lrcorner\mathbf{S}(\xi)\nonumber\\
&  =\frac{1}{2}(\boldsymbol{\gamma}^{\alpha}\wedge\boldsymbol{\gamma}_{\kappa
})\cdot\mathbf{S}(\xi)=\frac{1}{2}\mathbf{S}(\xi)\cdot(\boldsymbol{\gamma
}^{\alpha}\wedge\boldsymbol{\gamma}_{\kappa}). \label{lieb16}%
\end{align}

\end{remark}

From Eq.(\ref{lieb13}), the Lie derivative (when $\xi$ is Killing)
$\pounds _{\boldsymbol{\xi}}\boldsymbol{\gamma}^{\alpha}$ can be calculated in
two ways, using the usual definition by pullback or by the action of $u_{t}$.
Note, however that the action of $u_{t}$ is always orthogonal, regardless of
$\boldsymbol{\xi}$ be Killing. We will use this fact to give our geometric
motivated concept of Lie derivatives for Clifford and spinor fields.

\begin{remark}
It is very important to keep in mind that although the usual Lie derivative of
$\boldsymbol{\gamma}^{\alpha}$ in the direction of an arbitrary smooth Killing
vector field $\boldsymbol{\xi}$ is given by \emph{Eq.(\ref{lieb13}) }this does
not implies, of course that $\pounds _{\boldsymbol{\xi}}\boldsymbol{g}$ is
null for an arbitrary vector field. In fact $\pounds _{\boldsymbol{\xi}%
}\boldsymbol{g}=0$ defines a Killing vector field and
$\pounds _{\boldsymbol{\xi}}\boldsymbol{g}=0$ implies that there exists
$U_{t}\in\sec P_{\mathrm{Spin}_{1,3}^{e}}(M,\mathtt{g})$ such that in all
orders in $t$ \emph{(}not only in first order as in $Eq.(\ref{lieb13}))$ it
holds that
\begin{equation}
\mathrm{h}_{t}^{\ast}~\boldsymbol{\gamma}^{\alpha}=U_{t}^{-1}%
\boldsymbol{\gamma}^{\alpha}U_{t}=\mathbf{\Lambda}_{_{t}\beta}^{\mathbf{\alpha
}}\boldsymbol{\gamma}^{\beta}\label{nn1}%
\end{equation}
where for all $x\in M$, $\mathbf{\Lambda}_{_{t}\beta}^{\mathbf{\alpha}}%
(x)\in\mathrm{SO}_{1,3}^{e}$.
\end{remark}

\section{The Spinor Lie Derivative $\protect\overset{s}{\pounds }%
_{\boldsymbol{\xi}}$}

\subsection{Spinor Images of Clifford and Spinor Fields}

Given the spinorial frame $\mathbf{\Xi}_{u_{t}}(x)=(x,u_{t})$ in
$P_{\mathrm{Spin}_{1,3}^{e}}(M,\mathtt{g})$ we see that the basis
$\{\boldsymbol{\check{\gamma}}_{t}^{\alpha}\}$ of $P_{\mathrm{SO}_{1,3}^{e}%
}(M,\mathtt{g})$ such that
\begin{equation}
\boldsymbol{\check{\gamma}}_{t}^{\alpha}(x)=u_{t}^{-1}(x)\boldsymbol{\gamma
}^{\alpha}(x))u_{t}(x)=\Lambda_{t\beta}^{\alpha}(x)\boldsymbol{\gamma}^{\beta
}(x), \label{orthogonal}%
\end{equation}
is always orthonormal relative to $\mathtt{g}$. This suggests to define a
mapping $^{s}$\textrm{h}$_{t}$ \linebreak(associated with a one parameter
group of diffeomorphisms \textrm{h}$_{t}$ generated by a vector field
$\boldsymbol{\xi}$ acting on sections $%
%TCIMACRO{\tbigwedge \nolimits^{p}}%
%BeginExpansion
{\textstyle\bigwedge\nolimits^{p}}
%EndExpansion
T^{\ast}M,$ $\mathcal{C\ell(}M,\mathrm{g})$, $\mathcal{C}\ell_{\mathrm{Spin}%
_{1,3}^{e}}^{l}(M,\mathtt{g})$, $\mathcal{C}\ell_{\mathrm{Spin}_{1,3}^{e}}%
^{r}(M,\mathtt{g})$. With $x^{\prime}=\mathrm{h}_{t}(x)$ we start giving

\begin{definition}%
\begin{gather}
^{s}\mathrm{h}_{t}:\sec\mathcal{C\ell(}M,\mathrm{g})\hookleftarrow\sec%
%TCIMACRO{\tbigwedge \nolimits^{p}}%
%BeginExpansion
{\textstyle\bigwedge\nolimits^{p}}
%EndExpansion
T^{\ast}M\rightarrow\sec%
%TCIMACRO{\tbigwedge \nolimits^{p}}%
%BeginExpansion
{\textstyle\bigwedge\nolimits^{p}}
%EndExpansion
T^{\ast}M\hookrightarrow\sec\mathcal{C\ell(}M,\mathrm{g}),\nonumber\\
P(x^{\prime})\mapsto\check{P}_{t}(x)=\frac{1}{p!}P_{i_{1}\cdots i_{p}%
}(x^{\prime}(x))\boldsymbol{\check{\gamma}}_{t}^{i_{1}}(x)\cdots
\boldsymbol{\check{\gamma}}_{t}^{i_{1}}(x)\nonumber\\
=\frac{1}{p!}P_{i_{1}\cdots i_{p}}(x^{\prime}(x))u_{t}^{-1}\boldsymbol{\gamma
}^{i_{1}}(x)\cdots\boldsymbol{\gamma}^{i_{1}}(x)u_{t} \label{bg1}%
\end{gather}
with%
\begin{align}
P(x)  &  =\frac{1}{p!}P_{i_{1}\cdots i_{p}}(x)\boldsymbol{\gamma}^{i_{1}%
}(x)\cdots\boldsymbol{\gamma}^{i_{1}}(x)\neq\check{P}_{t}(x)\nonumber\\
P(x^{\prime})  &  =\frac{1}{p!}P_{i_{1}\cdots i_{p}}(x^{\prime}%
)\boldsymbol{\gamma}^{i_{1}}(x^{\prime})\cdots\boldsymbol{\gamma}^{i_{1}%
}(x^{\prime})
\end{align}
\emph{Eq.(\ref{bg1})} extends by linearity to all sections of $\mathcal{C\ell
(}M,\mathrm{g})$. Given any $\mathcal{C}\in\sec\mathcal{C\ell(}M,\mathrm{g})$
we will call $\mathcal{\check{C}}_{t}$ the spinor image of $\mathcal{C}$.
\end{definition}

\subsection{Spinor Derivative of Clifford and Spinor Fields}

\begin{definition}
The spinor Lie Derivative $\overset{\boldsymbol{s}}{\pounds }_{\boldsymbol{\xi
}}$of a Clifford field $\boldsymbol{C}=[(\Xi_{0},\mathcal{C})]$ (a section of
$\mathcal{C\ell(}M,\mathrm{g})$\emph{) in the direction of an arbitrary vector
field }$\boldsymbol{\xi}$\emph{ is}%
\begin{equation}
\overset{\boldsymbol{s}}{\pounds }_{\boldsymbol{\xi}}\boldsymbol{C=}\left.
\frac{d}{dt}\boldsymbol{\check{C}}_{t}\right\vert _{t=0} \label{glie}%
\end{equation}

\end{definition}

A trivial calculation gives
\begin{equation}
\overset{\boldsymbol{s}}{\pounds }_{\boldsymbol{\xi}}\boldsymbol{C}%
=\mathfrak{d}_{\boldsymbol{\xi}}\boldsymbol{C+}\frac{1}{4}[\boldsymbol{S}%
(\boldsymbol{\xi}),\boldsymbol{C}]. \label{GLIE1}%
\end{equation}

Given that a left DHSF\ $\varPsi$, a section of $\mathcal{C\ell}%
_{\mathrm{Spin}_{1,3}^{e}}^{l}\mathcal{(}M,\mathtt{g})$\ (respectively
$\varPhi$, a section of $\mathcal{C\ell}_{\mathrm{Spin}_{1,3}^{e}}%
^{r}\mathcal{(}M,\mathtt{g})$) can be written as%
\begin{align}
\boldsymbol{\psi}_{\Xi_{0}}  &  =\varPsi\mathbf{1}_{\Xi_{0}}^{r}%
,~~~~\boldsymbol{\phi}_{\Xi_{0}}=\mathbf{1}_{\Xi_{0}}^{l}\varPhi,\\
\boldsymbol{\psi}_{\Xi_{0}}\mathbf{1}_{\Xi_{0}}^{l}  &  =\varPsi\mathbf{1}%
_{\Xi_{0}}^{r}\mathbf{1}_{\Xi_{0}}^{l}=\varPsi1,~~~~\mathbf{1}_{\Xi_{0}}%
^{r}\boldsymbol{\phi}_{\Xi_{0}}=\mathbf{1}_{\Xi_{0}}^{r}\mathbf{1}_{\Xi_{0}%
}^{l}\varPhi=1\varPhi,\\
\ \varPsi  &  =\boldsymbol{\psi}_{\Xi_{0}}\mathbf{1}_{\Xi_{0}}^{l}%
=\boldsymbol{\psi}_{\Xi_{u}}\mathbf{1}_{\Xi u}^{l},~~~~\varPhi=\mathbf{1}%
_{\Xi_{0}}^{r}\boldsymbol{\phi}_{\Xi_{0}}=\mathbf{1}_{\Xi_{u}}^{r}%
\boldsymbol{\phi}_{\Xi_{u}}. \label{si3}%
\end{align}

Using Eq.(\ref{si3}) and thar $\boldsymbol{\psi}_{\Xi_{0}},\boldsymbol{\phi
}_{\Xi_{0}}\in\sec\mathcal{C\ell}^{0}\mathcal{(}M,\mathtt{g})$, where we know
how to act, we propose the following definition:

\begin{definition}
The spinor images of $\varPsi$ and $\varPhi$ are:%
\begin{align}
^{s}\mathrm{h}_{t}~\varPsi\boldsymbol{\circ}\mathrm{h}_{t}x  &  =\mathrm{~}%
^{s}\varPsi_{t}(x):=(\boldsymbol{u}_{t}^{-1}\boldsymbol{\psi}_{\Xi_{0}%
}(x)\boldsymbol{u}_{t})~^{s}\mathbf{1}_{t\Xi_{0}}^{l},\label{si4}\\
^{s}\mathrm{h}_{t}\varPhi\boldsymbol{\circ}\mathrm{h}_{t}x  &  =\mathrm{~}%
^{s}\varPhi_{t}\boldsymbol{(x)}:=~^{s}\mathbf{1}_{t\Xi_{0}}^{r}(\boldsymbol{u}%
_{t}^{-1}\boldsymbol{\phi}_{\Xi_{0}}(x)\boldsymbol{u}_{t}) \label{si5}%
\end{align}
and
\begin{equation}
^{s}\mathbf{1}_{t\Xi_{0}}^{l}:=\boldsymbol{u}_{t}^{-1}\mathbf{1}_{\Xi_{0}}%
^{l},~~~~~~^{s}\mathbf{1}_{t\Xi_{0}}^{r}=\mathbf{1}_{\Xi_{0}}^{r}u_{t}.
\label{si6}%
\end{equation}

\end{definition}

\begin{definition}
With these actions, we define:
\begin{align}
\overset{s}{\pounds }_{\boldsymbol{\xi}}\varPsi  &  \mathbf{:}=\frac{d}%
{dt}\left.  ^{s}\varPsi_{t}(x)\right\vert _{t=0},\nonumber\\
\overset{s}{\pounds }_{\boldsymbol{\xi}}\varPhi  &  :=\frac{d}{dt}\left.
^{s}\varPhi_{t}(x)\right\vert _{t=0}%
\end{align}

\end{definition}

The objects $^{s}\mathbf{\Psi}_{t},^{s}\varPsi_{t},^{s}\boldsymbol{\psi
}_{t\boldsymbol{\mathbf{\Xi}}_{0}},^{s}\mathbf{\Phi}_{t},^{s}\varPhi_{t}%
,^{s}\boldsymbol{\phi}_{t\boldsymbol{\mathbf{\Xi}}_{0}},^{s}\boldsymbol{C}%
_{t}$ ( sections of $I^{l}(M,\mathtt{g})$, $I^{r}(M,\mathtt{g})$,
$\mathcal{C\ell}_{\mathrm{Spin}_{1,3}^{e}}^{l}\mathcal{(}M,\mathtt{g})$,
$\mathcal{C\ell}_{\mathrm{Spin}_{1,3}^{e}}^{r}\mathcal{(}M,\mathtt{g})$) will
be referred in what follows as the \emph{spinor images} of the fields
$\mathbf{\Psi},\varPsi,\boldsymbol{\psi}_{\boldsymbol{\mathbf{\Xi}}_{0}%
},\mathbf{\Phi},\varPhi,\boldsymbol{\phi}_{\boldsymbol{\mathbf{\Xi}}_{0}%
},\boldsymbol{C}$.\smallskip

A trivial calculation gives
\begin{align}
\overset{\boldsymbol{s}}{\pounds _{\boldsymbol{\xi}}}\varPsi  &
=\mathfrak{d}_{\mathbf{\xi}}\varPsi+\frac{1}{4}\boldsymbol{S}(\xi
)\varPsi,\nonumber\\
\overset{\boldsymbol{s}}{\pounds _{\boldsymbol{\xi}}}\varPhi  &
=\mathfrak{d}_{\mathbf{\xi}}\varPhi-\varPhi\frac{1}{4}\boldsymbol{S}(\xi)
\label{bglie1}%
\end{align}
and we observe that it is%
\begin{equation}
\overset{s}{\pounds }_{\boldsymbol{\xi}}\mathbf{1}_{\Xi_{0}}^{l}=\frac{1}%
{4}\mathbf{S}(\boldsymbol{\xi})\mathbf{1}_{\Xi_{0}}^{l}%
,~~~~\overset{s}{\pounds }_{\boldsymbol{\xi}}\mathbf{1}_{\Xi_{0}}^{r}%
=-\frac{1}{4}\mathbf{1}_{\Xi_{0}}^{r}\mathbf{S}(\boldsymbol{\xi}).
\label{liebh}%
\end{equation}

\begin{remark}
In the Clifford bundle in the basis $\mathbf{\Xi}_{0}$, $\psi_{\mathbf{\Xi
}_{0}}\in\sec\mathcal{C}\ell(M,\mathtt{g})$ is the representative of $\Psi
\in\sec\mathcal{C}\ell_{\mathrm{Spin}_{1,3}^{e}}^{\ell}(M,\mathtt{g})$ and if
we calculated its spinor Lie derivative as a section of\ $\mathcal{C}%
\ell(M,\mathtt{g})$ we should get, of course%
\begin{equation}
\overset{s}{\pounds }_{\boldsymbol{\xi}}\psi_{\mathbf{\Xi}_{0}}=\mathfrak{d}%
_{\boldsymbol{\xi}}\psi_{\mathbf{\Xi}_{0}}+\frac{1}{4}[L(\xi\boldsymbol{)}%
+d\xi,\psi_{\mathbf{\Xi}_{0}}]. \label{lie9}%
\end{equation}

This does not mimics the spinor Lie derivative of\ a DHSF $\Psi$. Since one of
the main reasons to introduce representatives in the Clifford bundle of
Dirac-Hestenes spinor fields is to have an easy computation tool when using
these representatives together with other Clifford fields we will agree to
take as the Lie derivative of $\psi_{\mathbf{\Xi}_{0}}$ an effective Lie
derivative denoted $\overset{(s)}{\pounds }_{\boldsymbol{\xi}}\psi
_{\mathbf{\Xi}_{0}}$ where the pullback of $\psi_{\mathbf{\Xi}_{0}}$ is the
formula given by \emph{Eq.(\ref{lie3bb})}. Thus,
\begin{equation}
\overset{(s)}{\pounds }_{\boldsymbol{\xi}}\psi_{\mathbf{\Xi}_{0}}%
=\mathfrak{d}_{\boldsymbol{\xi}}\psi_{\mathbf{\Xi}_{0}}+\frac{1}{4}%
L(\xi\boldsymbol{)}\psi_{\mathbf{\Xi}_{0}}+\frac{1}{4}d\xi\psi_{\mathbf{\Xi
}_{0}} \label{lie10}%
\end{equation}
We then write for $\mathcal{C}\in\sec\mathcal{C}\ell(M,\mathtt{g})$ and
$\psi_{\mathbf{\Xi}_{0}}$ as just defined
\begin{equation}
\overset{(s)}{\pounds }_{\boldsymbol{\xi}}(\mathcal{C}\psi_{\mathbf{\Xi}_{0}%
})=(\overset{s}{\pounds }_{\boldsymbol{\xi}}\mathcal{C)}\psi_{\mathbf{\Xi}%
_{0}})+\mathcal{C}(\overset{(s)}{\pounds }_{\boldsymbol{\xi}}\psi
_{\mathbf{\Xi}_{0}}). \label{lie11}%
\end{equation}
\ 
\end{remark}

\begin{remark}
An analogous concept to $\overset{(s)}{\pounds }_{\boldsymbol{\xi}}$ has been
introduced in \emph{\cite{rc2007}} for the covariant derivative of
representatives in the Clifford bundle of Dirac-Hestenes spinor fields and we
recall that for $\mathcal{C}\in\sec\mathcal{C}\ell(M,\mathtt{g})$ and
$\psi_{\mathbf{\Xi}_{0}}$ as above defined we have%
\begin{align}
\overset{(s)}{\boldsymbol{D}}_{\boldsymbol{\xi}}(\mathcal{C}\psi_{\mathbf{\Xi
}_{0}})  &  =(\boldsymbol{D}_{\boldsymbol{\xi}}\mathcal{C)}\psi_{\mathbf{\Xi
}_{0}}+\mathcal{C}(\overset{(s)}{\boldsymbol{D}}_{\boldsymbol{\xi}}%
\psi_{\mathbf{\Xi}_{0}}),\nonumber\\
\boldsymbol{D}_{\boldsymbol{\xi}}\mathcal{C}  &  =\mathfrak{d}_{\xi
}\mathcal{C+}\frac{1}{2}[\omega_{\xi},\mathcal{C}],\nonumber\\
\overset{(s)}{\boldsymbol{D}}_{\boldsymbol{\xi}}\psi_{\mathbf{\Xi}_{0}}  &
=\mathfrak{d}_{\xi}\psi_{\mathbf{\Xi}_{0}}\mathcal{+}\frac{1}{2}\omega_{\xi
}\psi_{\mathbf{\Xi}_{0}}. \label{cd1234}%
\end{align}
with
\begin{equation}
\omega_{\xi}:=\frac{1}{2}\xi^{\kappa}\omega_{\alpha\kappa\beta}%
\boldsymbol{\gamma}^{\alpha}\boldsymbol{\gamma}^{\beta}. \label{c2f}%
\end{equation}
called the \textquotedblleft connection $2$-form\textquotedblright.
Henceforth, to simplify the notation, the covariant derivative acting in a
representative\ in the Clifford bundle of a DHSF will be written as%
\[
\boldsymbol{D}_{\boldsymbol{\xi}}^{s}\psi_{\mathbf{\Xi}_{0}}=\mathfrak{d}%
_{\xi}\psi_{\mathbf{\Xi}_{0}}\mathcal{+}\frac{1}{2}\omega_{\xi}\psi
_{\mathbf{\Xi}_{0}}%
\]
and we will write also $\overset{s}{\pounds }_{\boldsymbol{\xi}}%
\psi_{\mathbf{\Xi}_{0}}$ (given by \emph{Eq.(\ref{lie11}))} instead of
$\overset{(s)}{\pounds }_{\boldsymbol{\xi}}\psi_{\mathbf{\Xi}_{0}}.$
\end{remark}

\begin{remark}
One can easily verify that with this agreement we have a perfectly consistent
formalism. Indeed, recalling that the spinor bundles are modules over
$\mathcal{C\ell(}M,\mathtt{g})$ and that any section $\boldsymbol{C}$ of
$\mathcal{C\ell(}M,\mathtt{g})$ \emph{(see Eq.(\ref{PAIRING}))} can be written
as the product of a section $\Psi$ of $\mathcal{C\ell}_{\mathrm{Spin}%
_{1,3}^{e}}^{l}\mathcal{(}M,\mathtt{g})$ by a section $\Phi$ of
$\mathcal{C\ell}_{\mathrm{Spin}_{1,3}^{e}}^{r}\mathcal{(}M,\mathtt{g})$, i.e.,
$\boldsymbol{C}=\Psi\Phi$ we immediately verify that the operator
$\overset{s}{\pounds }_{\boldsymbol{\xi}}$ satisfies when applied to Clifford
and spinor fields the Leibniz rule, i.e.,
\begin{align}
\overset{s}{\pounds }_{\boldsymbol{\xi}}(\varPsi\varPhi)  &
=(\overset{s}{\pounds }_{\boldsymbol{\xi}}%
\varPsi)\varPhi+\varPsi(\overset{s}{\pounds }_{\boldsymbol{\xi}}%
\varPhi),\label{lieb16c}\\
\overset{s}{\pounds }_{\boldsymbol{\xi}}(\boldsymbol{C}\varPsi)  &
=(\overset{s}{\pounds }_{\boldsymbol{\xi}}\boldsymbol{C}%
)\varPsi+\boldsymbol{C}(\overset{s}{\pounds }_{\boldsymbol{\xi}}%
\varPsi),\label{liebc1}\\
\overset{s}{\pounds }_{\boldsymbol{\xi}}(\varPhi\boldsymbol{C})  &
=(\overset{s}{\pounds }_{\boldsymbol{\xi}}\boldsymbol{C}%
)\varPhi+\boldsymbol{C}(\overset{s}{\pounds }_{\boldsymbol{\xi}}\varPhi).
\label{liebc2}%
\end{align}

\end{remark}

\section{The Spinor Lie Derivative Written in Terms of Covariant Derivatives}

Recalling Eq.(\ref{lie3b12d})
\begin{equation}
\frac{1}{4}L(\xi\boldsymbol{)}\psi_{\mathbf{\Xi}_{0}}=\frac{1}{2}\omega_{\xi
}\psi_{\mathbf{\Xi}_{0}}%
\end{equation}
and that%
\begin{equation}
d\xi=\gamma^{\alpha}\wedge(\boldsymbol{D}_{e_{\alpha}}\xi)=\frac{1}%
{2}(\boldsymbol{D}_{\alpha}\xi_{\beta}-\boldsymbol{D}_{\beta}\xi_{\alpha
})\boldsymbol{\gamma}^{\alpha}\boldsymbol{\gamma}^{\beta}%
\end{equation}
we see that Eq.(\ref{cd1234}) gives%

\begin{equation}
\overset{s}{\pounds }_{\boldsymbol{\xi}}\psi_{\mathbf{\Xi}_{0}}%
=~\boldsymbol{D}_{\boldsymbol{\xi}}\psi_{\mathbf{\Xi}_{0}}-\frac{1}%
{8}(\boldsymbol{D}_{\alpha}\xi_{\beta}-\boldsymbol{D}_{\beta}\xi_{\alpha
})\boldsymbol{\gamma}^{\alpha}\boldsymbol{\gamma}^{\beta}\psi_{\mathbf{\Xi
}_{0}}. \label{liedd}%
\end{equation}

\subsection{Calculation of $\protect\overset{s}{\pounds }_{\boldsymbol{\xi}%
}\psi_{\mathbf{\Xi}_{0}}$ in Coordinates}

We show now that the spinor Lie derivative of spinor fields coincides with the
one first introduced by Lichnerowicz \cite{lichnerowicz}. To see this, we
evaluate the spinor Lie derivative of a spinor field introducing coordinates
$\{x^{\mu}\}$ for $\mathcal{U}\subset M$. We write
\begin{equation}
\boldsymbol{\gamma}^{\mathbf{\alpha}}=h_{\mu}^{\mathbf{\alpha}}dx^{\mu}%
\end{equation}
and with $\boldsymbol{D}$ the Levi-Civita connection of $\boldsymbol{g}$ we
write as usual%
\begin{equation}
\boldsymbol{D}_{\partial_{\mu}}\boldsymbol{\gamma}^{\beta}=-\omega_{\cdot
\mu\mathbf{\alpha}}^{\mathbf{\beta}\cdot\cdot}\boldsymbol{\gamma
}^{\mathbf{\alpha}},~~~~\boldsymbol{D}_{\partial_{\mu}}dx^{\nu}=-\Gamma
_{\cdot\mu\tau}^{\nu\cdot\cdot}dx^{\tau}.
\end{equation}
Thus, as well known%
\begin{equation}
\partial_{\mu}h_{\nu}^{\mathbf{\alpha}}+\omega_{\cdot\mu\mathbf{\beta}%
}^{\mathbf{\alpha}\cdot\cdot}h_{\nu}^{\mathbf{\beta}}-h_{\sigma}%
^{\mathbf{\alpha}}\Gamma_{\mu\cdot\nu}^{\sigma\cdot\cdot}=0,
\end{equation}
from where we get
\begin{equation}
\omega_{\mathbf{\alpha}\mu\mathbf{\beta}}=-(\partial_{\mu}h_{\mathbf{\alpha
}\nu})h_{\mathbf{\beta}}^{\nu}+\Gamma_{\mu\mathbf{\alpha\beta}}. \label{upa}%
\end{equation}
Take notice that in writing Eq.(\ref{upa}) we used in agreement with the
original definition of the Christofell symbols that it is \emph{licit} to
write in a coordinate basis (as some authors do, e.g., \cite{kn,lovrund})%
\[
\Gamma_{\cdot\mu\nu}^{\rho\cdot\cdot}=\Gamma_{\mu\cdot\nu}^{\cdot\rho\cdot
}=g^{\rho\sigma}\Gamma_{\mu\sigma\nu}:=\frac{1}{2}\left(  \frac{\partial
g_{\nu\sigma}}{\partial x^{\mu}}+\frac{\partial g_{\sigma\mu}}{\partial
x^{\nu}}-\frac{\partial g_{\mu\nu}}{\partial x^{\sigma}}\right)  .
\]
So, we get
\begin{align}
\overset{s}{\pounds }_{\boldsymbol{\xi}}\psi_{\mathbf{\Xi}_{0}}  &
=\mathfrak{d}_{\xi}\psi_{\mathbf{\Xi}_{0}}-\frac{1}{4}h_{\mathbf{\beta}}^{\nu
}\{\xi^{\mu}\partial_{\mu}h_{\mathbf{\alpha}\nu}+\xi^{\mu}\Gamma
_{\mathbf{\alpha}\mu\mathbf{\beta}}-(\partial_{\nu}\xi_{\mu})h_{\alpha}^{\mu
}\}\boldsymbol{\gamma}^{\alpha}\wedge\boldsymbol{\gamma}^{\beta}%
\psi_{\mathbf{\Xi}_{0}}\nonumber\\
&  =\mathfrak{d}_{\xi}\psi_{\mathbf{\Xi}_{0}}-\frac{1}{4}h_{\beta}^{\nu}%
\{\xi^{\mu}\partial_{\mu}h_{\alpha\nu}-(\partial_{\nu}\xi_{\mu})h_{\alpha
}^{\mu}\}\boldsymbol{\gamma}^{\alpha}\wedge\boldsymbol{\gamma}^{\beta}%
\psi_{\mathbf{\Xi}_{0}}\nonumber\\
&  -\frac{1}{4}\{\xi^{\rho}\Gamma_{\rho\mu\nu}\}dx^{\mu}\wedge dx^{\nu}%
\psi_{\mathbf{\Xi}_{0}}\nonumber\\
&  =\mathfrak{d}_{\xi}\psi_{\mathbf{\Xi}_{0}}-\frac{1}{4}h_{\beta}^{\nu}%
\{\xi^{\mu}\partial_{\mu}h_{\alpha\nu}-(\partial_{\nu}\xi_{\mu})h_{\alpha
}^{\mu}\}\boldsymbol{\gamma}^{\alpha}\wedge\boldsymbol{\gamma}^{\beta}%
\psi_{\mathbf{\Xi}_{0}} \label{liedhsf}%
\end{align}
where the last term in the second line of Eq.(\ref{liedhsf})\ is null because
$\Gamma_{\rho\mu\nu}=\Gamma_{\rho\nu\mu}$.

\subsection{Spinor Lie Derivative of Covariant Dirac Spinor Fields}

\begin{definition}
For completeness we recall that the spinor Lie derivative of a covariant Dirac
spinor field\emph{\footnote{The column spinor fields used in physical
textbooks.}} $\boldsymbol{\Psi}\in\sec P_{\mathrm{Spin}_{1,3}^{e}%
}(M,\mathtt{g})_{\times_{\mu}}\mathbb{C}^{4}$, with $\mu$ the $D^{(1/2.0)}%
\oplus D^{(0.1/2)}$ representation of \textrm{Sl}$(2,C)\simeq\mathrm{Spin}%
_{1,3}^{e}$ is%
\begin{align}
\overset{s}{\pounds }_{\boldsymbol{\xi}}\boldsymbol{\Psi}  &  =\frac{d}%
{dt}\left.  ^{s}\mathbf{\Psi}_{t}(x)\right\vert _{t=0}\nonumber\\
&  =\mathfrak{d}_{\boldsymbol{\xi}}\boldsymbol{\Psi}-\frac{1}{4}%
d\xi\mathbf{\Psi}+\frac{1}{4}\xi^{\kappa}\omega_{\alpha\kappa\beta
}\underline{\gamma}^{\alpha}\underline{\gamma}^{\beta} \label{LIE9a}%
\end{align}
with the $\underline{\gamma}^{\alpha\prime}$s being matrix representations of
the $\boldsymbol{\gamma}^{\alpha}$'s. Of course,
\begin{equation}
\overset{s}{\pounds }_{\boldsymbol{\xi}}\mathbf{\Psi}=\boldsymbol{D}%
_{\boldsymbol{\xi}\text{ }}\mathbf{\Psi}-\frac{1}{8}\left(  \boldsymbol{D}%
_{\alpha}\xi_{\beta}-\boldsymbol{D}_{\beta}\xi_{\alpha}\right)
\underline{\gamma}^{\alpha}\underline{\gamma}^{\beta}\mathbf{\Psi}, \label{LC}%
\end{equation}

\end{definition}

\section{The Spinor Lie derivative of the Metric Field}

We want to extend the spinor Lie derivative that we just defined for Clifford
and spinor fields to sections of the tensor bundle. The general case will not
be discusssed today. Here we give the

\begin{definition}
The spinor Lie derivative of the metric field $\boldsymbol{g}=\eta
_{\mathbf{\alpha\beta}}\boldsymbol{\gamma}^{\mathbf{\alpha}}\otimes
\boldsymbol{\gamma}^{\mathbf{\beta}}$ in the direction of the arbitrary vector
field $\boldsymbol{\xi}$ is%
\begin{equation}
\overset{s}{\pounds }_{\boldsymbol{\xi}}\boldsymbol{g}(x)=\lim_{t\rightarrow
0}\frac{\boldsymbol{\check{g}}_{t}(x)-\boldsymbol{g}(x)}{t}%
,~~~\boldsymbol{\check{g}}_{t}(x):=\eta_{\mathbf{\alpha\beta}}%
\boldsymbol{\check{\gamma}}_{t}^{\alpha}(x)\otimes\boldsymbol{\check{\gamma}%
}_{t}^{\beta}(x)
\end{equation}

\end{definition}

\begin{proposition}
\textbf{ }For any vector field $\boldsymbol{\xi}\overset{s}{\text{, it is
}\pounds }_{\boldsymbol{\xi}}\boldsymbol{g}=0.$
\end{proposition}

\begin{proof}
Indeed we have%
\begin{gather}
\overset{s}{\pounds }_{\boldsymbol{\xi}}\boldsymbol{g}(x)=\lim_{t\rightarrow
0}\frac{\eta_{\alpha\beta}\boldsymbol{\check{\gamma}}_{t}^{\alpha}%
(x)\otimes\boldsymbol{\check{\gamma}}_{t}^{\beta}(x)-\eta_{\alpha\beta
}\boldsymbol{\gamma}_{t}^{\alpha}(x)\otimes\boldsymbol{\gamma}_{t}^{\beta}%
(x)}{t}\nonumber\\
=\lim_{t\rightarrow0}\frac{\eta_{\alpha\beta}(u_{t}^{-1}(x)\boldsymbol{\gamma
}^{\alpha}(x)u_{t}(x))\otimes(u_{t}^{-1}(x)\boldsymbol{\gamma}^{\beta}%
(x)u_{t}(x-\eta_{\alpha\beta}\boldsymbol{\gamma}_{t}^{\alpha}(x)\otimes
\boldsymbol{\gamma}_{t}^{\beta}(x)))}{t}\nonumber\\
=\lim_{t\rightarrow0}\frac{\eta_{\alpha\beta}\Lambda_{t\kappa}^{\alpha
}(x)\boldsymbol{\gamma}^{\kappa}(x)\otimes\Lambda_{t\iota}^{\beta
}(x)\boldsymbol{\gamma}^{\iota}(x)-\eta_{\alpha\beta}\boldsymbol{\gamma}%
_{t}^{\alpha}(x)\otimes\boldsymbol{\gamma}_{t}^{\beta}(x)))}{t}\nonumber\\
=\lim_{t\rightarrow0}\frac{\boldsymbol{g}(x)-\boldsymbol{g}(x)}{t}=0.
\label{lzerob}%
\end{gather}

\end{proof}

\begin{remark}
So, technically speaking $\overset{s}{\pounds }_{\boldsymbol{\xi}}$ is a
perfectly well defined operator acting in the Clifford and spin-Clifford
bundles having the nice properties exhibited above. It defines a derivation in
the Clifford bundle since it does pass to the quotient $\boldsymbol{\tau}%
M/I=$\ $\mathcal{C\ell(}M,\mathrm{g})$, where $I$ is the bilateral ideal
generated by elements of the form $(a\otimes b+b\otimes a-2\mathrm{g}(a,b))$
with $a,b\in\sec%
%TCIMACRO{\tbigwedge \nolimits^{1}}%
%BeginExpansion
{\textstyle\bigwedge\nolimits^{1}}
%EndExpansion
T^{\ast}M$.
\end{remark}

Of course, when $\boldsymbol{\xi}$ is a Killing vector field we have
$\pounds _{\boldsymbol{\xi}}\mathcal{C}=\overset{s}{\pounds }%
\mathcal{_{\boldsymbol{\xi}}C}$

\begin{remark}
A definition of Lie derivative of tensor and spinor fields such that it always
annihilates $\boldsymbol{g}$ has be given by Bourguignon \& Gauduchon
\emph{\cite{bg1992}} using a very different method.
\end{remark}

\section{Other Definitions of Lie Derivatives of Spinor Fields.}

Once again we recall that the definition of the \emph{spinor} Lie derivative
of spinor fields generated by Killing vector fields has first given by
Lichnerowicz \cite{lichnerowicz} and taken valid (as a definition) for
arbitrary diffeomorphisms generated by arbitrary vector fields by Kosmann
\cite{kosmann1,kosmann2,kosmann3}. A \textquotedblleft
justification\textquotedblright\ of Kosmann's formula for the case where
$\boldsymbol{\xi}$ is a Killing vector field is given, e.g., in \cite{bt1987}.
There, it is imposed that the Lie derivative be a \emph{derivation} on the
Clifford bundle, by passing to the quotient the action of the Lie derivative
on $\boldsymbol{\tau}M/I=$\ $\mathcal{C\ell(}M,\mathrm{g})$ which necessarily
implies that $\boldsymbol{\xi}$ must be a Killing vector field. The Lie
derivative is then obtained using an analogy with the concept of covariant
derivative in the following way. First, one recall that the covariant
derivative of a Clifford field $\mathcal{C}\in\sec\mathcal{C\ell(}%
M,\mathrm{g})$ can be written as%
\begin{equation}
\boldsymbol{D}_{\boldsymbol{\xi}}\mathcal{C}=\mathfrak{d}_{\boldsymbol{\xi}%
}\mathcal{C+}\frac{1}{2}[\omega_{\xi},\mathcal{C}] \label{cdc}%
\end{equation}
and the covariant derivative of a representative\ in the Clifford bundle of a
DHSF can be written as%
\begin{equation}
\overset{s}{\boldsymbol{D}}_{\boldsymbol{\xi}}\psi_{\mathbf{\Xi}_{0}%
}=\mathfrak{d}_{\boldsymbol{\xi}}\psi_{\mathbf{\Xi}_{0}}+\frac{1}{2}%
\omega_{\boldsymbol{\xi}}\psi_{\mathbf{\Xi}_{0}}.
\end{equation}
Next, showing that for a Killing vector field $\boldsymbol{\xi}$ and $X\in\sec%
%TCIMACRO{\tbigwedge \nolimits^{1}}%
%BeginExpansion
{\textstyle\bigwedge\nolimits^{1}}
%EndExpansion
T^{\ast}M\hookrightarrow\sec\mathcal{C\ell(}M,\mathrm{g})$ the standard Lie
derivative of $X$ is
\begin{align}
\pounds _{\boldsymbol{\xi}}X  &  =\mathfrak{d}_{\boldsymbol{\xi}}X+\frac{1}%
{4}[L(\xi)+d\xi,X]\\
&  =\boldsymbol{D}_{\boldsymbol{\xi}}X+\frac{1}{4}[d\xi,X]
\end{align}
it is \emph{postulated} that.%
\begin{align}
\pounds _{\boldsymbol{\xi}}\psi_{\mathbf{\Xi}_{0}}  &  :=\mathfrak{d}%
_{\boldsymbol{\xi}}\psi_{\mathbf{\Xi}_{0}}\mathcal{+}\frac{1}{4}(L(\xi
)+d\xi)\psi_{\mathbf{\Xi}_{0}}\\
&  =\overset{s}{\boldsymbol{D}}_{\boldsymbol{\xi}}\psi_{\mathbf{\Xi}_{0}%
}+\frac{1}{4}d\xi\psi_{\mathbf{\Xi}_{0}}%
\end{align}
which is the Kosmann's formula.

\begin{remark}
We emphasize that we get the above formulas with a very different
procedure\footnote{Without introducing at starting the use of Levi-Civita
connections, something that seems unjustificable if we want a meaningful
notion of Lie derivative.}, namely by finding a geometrical motivated
definition for the image of Clifford and spinor fields generated by
one-parameter groups of diffeomorphisms associated to an arbitrary smooth
vector field $\boldsymbol{\xi}$.
\end{remark}

\begin{remark}
\label{tensorcase}We also mention that\ in \emph{\cite{hv2000}} a Lie
derivative of a spinor field is also defined using analogy with the covariant
derivative and a formula is obtained similar to the formulas that we get for
$\overset{s}{\pounds }$ but with an extra term, which authors claim to be
necessary in order to have agreement with the Lie derivative of general tensor
fields with the ones obtained from the Lie derivative of general tensors
fields represented by a tensor product of spinor fields. As already mentioned
above this important issue will be discussed in another publication. Anyway,
we emphasize that our definition of $\overset{s}{\pounds }$ seems perfectly
consistent with the Clifford and spin-Clifford bundles formalism.
\end{remark}

\begin{remark}
It is worth to mention that an equation similar\emph{ Eq.(\ref{liedd})} has
also been obtained in \emph{\cite{F}} using the general concept of Lie
differentiation in the elegant theory of gauge natural bundles.\ The theory of
Lie differentiation in gauge natural bundles\ was originally introduced
by\emph{ \cite{eck}} and developed by Klok\'{a}\v{r} and
collaborators\emph{\footnote{See \cite{KMS}\emph{.}}}. It is reviewed with
emphasis in physical applications in \emph{\cite{F1,gm1,gm2}}. Authors
\emph{\cite{gm1,gm2}} claim that \emph{\cite{F}} succeeded in given a
geometrical meaning for the Kosmann definition, but the case is that what has
been done there was to postulated a particular lifting of the vector field
$\boldsymbol{\xi}\in\sec TM$ to the tangent bundle to $P_{\mathrm{Spin}%
_{1,3}^{e}}(M,\boldsymbol{g})$\ such that the definition of\ \emph{(}%
generalized\emph{) }Lie derivative of a spinor field results in Kosmann's
formula. This is a sophisticated way to get the same result we get using a
very simple and intuitive path.
\end{remark}

\begin{remark}
\label{grav}To obtain the Euler-Lagrange equations from the principle of
stationary action for a system consisting of a spinor field, the
electromagnetic and gravitational field\emph{\footnote{With the gravitational
field represented by the cotetrad fields $\{\boldsymbol{\gamma}^{\alpha}\}$.}}
implies in giving a clear definition of what we mean by the variation of these
fields. If the variations of the Clifford fields\ representing the
electromagnetic field and of the spinor fields representing matter if given by
$\overset{s}{\pounds }_{\boldsymbol{\xi}}$ we will have as a consequence that
the metric $\boldsymbol{g}$ defined by the cotetrad fields will have null
variation when the cotetrad fields are varied.
\end{remark}

\begin{remark}
\label{grav1}We recall here that when we represent the gravitational field by
the cotetrad fields $\boldsymbol{\gamma}^{\alpha}$ in a Riemann-Cartan theory
\emph{(}see \emph{\cite{rc2007})} we need, in order to obtain covariant
conservation laws for the matter and electromagnetic fields to make use of
vertical \emph{(}$\boldsymbol{\delta}_{\mathbf{v}}\boldsymbol{\gamma}^{\alpha
}=\Lambda_{\beta}^{\alpha}\boldsymbol{\gamma}^{\beta}$, with $\Lambda_{\beta
}^{\alpha}$ a local Lorentz rotation\emph{) }and horizontal \emph{(}%
$\boldsymbol{\delta}_{\mathbf{h}}\boldsymbol{\gamma}^{\alpha}=-$
$\pounds _{\boldsymbol{\varkappa}}\boldsymbol{\gamma}^{\alpha}$) variations.
However, existence of genuine \emph{(}not the covariant ones\emph{)
}conservation laws requires the existence of appropriated Killing vector
fields and consistency of the formalism requires in that case that%
\begin{equation}
\boldsymbol{\delta}_{\mathbf{v}}\boldsymbol{\gamma}^{\alpha}=-\Lambda_{\beta
}^{\alpha}\boldsymbol{\gamma}^{\beta}=-\pounds _{\boldsymbol{\varkappa}%
}\boldsymbol{\gamma}^{\alpha}.
\end{equation}
These constrained variations of the $\boldsymbol{\gamma}^{\alpha}$ have been
used in the theory of the gravitational field developed in \emph{\cite{rc2007}%
} and with improvements in \emph{\cite{rod2012}}. Now, taking into account
that it is $\overset{\boldsymbol{s}}{\pounds _{\boldsymbol{\xi}}%
}\boldsymbol{g}=0$, it is the case that $\overset{\boldsymbol{s}%
}{\pounds _{\boldsymbol{\xi}}}\boldsymbol{\gamma}^{\alpha}=\Lambda_{\beta
}^{\alpha}\boldsymbol{\gamma}^{\beta}$ we see that a consistent Lagrangian
formalism for fields represented by Clifford fields \emph{(}this including the
representatives in the Clifford bundle of Dirac-Hestenes spinor fields\emph{)
}may be based on taken variations of the fields $\phi$ entering the Lagrangian
density as being $\overset{\boldsymbol{s}}{\pounds _{\boldsymbol{\xi}}}\phi$.
This will be discussed in another publication.
\end{remark}

\section{Conclusions}

In this paper we claim to have given a geometrical motivated definition for a
Lie derivative of spinor fields in a Lorentzian structure $(M,\boldsymbol{g})$
by finding an appropriated image for Clifford and spinor fields under a
diffeomorphism generated by an arbitrary vector field $\boldsymbol{\xi}$ . We
called such operator the \emph{spinor Lie derivative}, denoted
$\overset{s}{\pounds }_{\boldsymbol{\xi}}$which is such that
$\overset{s}{\pounds }_{\boldsymbol{\xi}}\boldsymbol{g}=0$ for arbitrary
vector field $\boldsymbol{\xi}$ . We compared our definitions and results with
the many others appearing in literature on the subject.

\end{document}